\documentclass[11pt,a4paper]{article}

\usepackage{amsmath,amssymb}

\usepackage{amsthm}
\usepackage{hyperref}
\usepackage{mathtools}
\usepackage{enumitem}
\usepackage{bbm}
\usepackage{graphicx,graphics,float}
\usepackage{caption}
\usepackage{subcaption}
\usepackage[left=1in, right=1in, top=1in, bottom=1in, includefoot, headheight=15pt]{geometry}
\usepackage{stmaryrd}
\usepackage{natbib}%for bibliography style author/year

\usepackage{color}
\definecolor{darkgreen}{rgb}{0, .5, 0}
\definecolor{darkred}{rgb}{.5, 0, 0}
\definecolor{amethyst}{rgb}{0.6, 0.4, 0.8}
\definecolor{blue-violet}{rgb}{0.54, 0.17, 0.89}
\definecolor{orange(colorwheel)}{rgb}{1.0, 0.5, 0.0}

\usepackage{setspace}

\theoremstyle{plain}
\newtheorem{theorem}{Theorem}[section]

\theoremstyle{definition} %%

\usepackage{mathrsfs}

\newcommand{\E}{{\mathbb{E}}}

\newcommand{\Var}{{\mathbb{V}ar}}
\newcommand{\Cov}{{\mathbb{C}ov}}

\newcommand{\classltwoone}{\ensuremath{\mathbb{L}_{T}^{2, 1}}}

\newcommand{\classloneone}{\ensuremath{\mathbb{L}_{T}^{1, 1}}}

\newcommand{\classlfourone}{\ensuremath{\mathbb{L}_{T}^{4, 1}}}

\newcommand{\condE}[1]{\ensuremath{\mathbb{E}[#1 | \mathcal{F}_{k}]}}
\newcommand{\condEplus}[1]{\ensuremath{\mathbb{E}[#1 | \mathcal{F}_{k+1}]}}

\newcommand{\condVar}[1]{\ensuremath{\mathbb{V}ar(#1 | \mathcal{F}_{k})}}
\newcommand{\condCov}[1]{\ensuremath{\mathbb{C}ov(#1 | \mathcal{F}_{k})}}
\newcommand{\Pas}{\ensuremath{\mathbb{P}\text{ -- a.s.}}}

\newcommand{\e}{\ensuremath{e}}

\ifx\prop\undefined
\newtheorem{prop}{Proposition}[section]
\fi
\newtheorem{theo}[prop]{Theorem}
\newtheorem{lem}[prop]{Lemma}
\newtheorem{cor}[prop]{Corollary}
\newtheorem{defi}[prop]{Definition}
\newtheorem{remark}[theorem]{Remark}

\def\P{{\mathbb P}}

\providecommand{\abs}[1]{\ensuremath{\left\lvert#1\right\rvert}}

\providecommand{\norm}[1]{\ensuremath{\lVert#1\rVert}}

\providecommand{\diag}{\ensuremath{{\rm diag}}}

\begin{document}
\singlespacing %
\title{\LARGE\bf {Local risk-minimization with multiple assets under illiquidity with applications in energy markets}}%
\author{Panagiotis Christodoulou\thanks{Department of Mathematics, University of Munich, Theresienstra\ss{}e 39, 80333 Munich, Germany. Email: christo@math.lmu.de} \and Nils Detering\thanks{Department of Statistics \& Applied Probability, University of California, Santa Barbara, CA 93106-3110, USA. Email: detering@pstat.ucsb.edu}\and Thilo Meyer-Brandis\thanks{Department of Mathematics, University of Munich, Theresienstra\ss{}e 39, 80333 Munich, Germany. Email: meyerbr@math.lmu.de}%
 }
 \date{February 28, 2018}
\maketitle

\begin{abstract}
We propose a hedging approach for general contingent claims when liquidity is
a concern and trading is subject to transaction cost. Multiple assets
with different liquidity levels are available for hedging. Our
risk criterion targets a tradeoff between minimizing the
risk against fluctuations in the stock price and incurring low
liquidity costs. Following \cite{cetin.jarrow.protter:2004} we work in an
arbitrage-free setting assuming a supply curve for each asset. In
discrete time, following the ideas in Schweizer et. al. \citep{schweizer:1988,lamberton.pham.schweizer:1998} we prove the
existence of a locally risk-minimizing strategy under mild conditions
on the price process. Under stochastic and time-dependent liquidity risk we give a closed-form solution for an optimal strategy in the case of a linear supply curve model. 
%Our results extend the work of Schweizer \cite{schweizer:1988}\nils{da haben wir doch keinen Zugriff drauf oder? Gibt es nicht ein paper dazu?} \panos{Leider, da haben wir keinen Zugriff und darein Schweizer macht das sowohl in diskreter Zeit als auch in stetiger Zeit (steht in seiner Abstrakt). Ein paper dazu von Schweizer habe ich nicht gefunden. Nur das paper Schweizer.1993.Option Hedging for semimartingales. Das ist in stetiger Zeit und basiert auf \cite{schweizer:1988} welches er erweitert.} in two
%directions. First by considering a multidimensional asset price.
%Second by accounting for liquidity costs.
Finally we show how our hedging method can be applied in energy
markets where futures with different maturities are available for
trading. The futures closest to their delivery period are usually the
most liquid but depending on the contingent claim not necessary
optimal in terms of hedging. In a simulation study we investigate
this tradeoff and compare the resulting hedge strategies with the
classical ones. 
\end{abstract}


%\medskip
%\noindent\textit{2010 Mathematics Subject Classification:} XXXXX


%Introduction
\section{Introduction}\label{sec:intro}
%
%\panos{Panos comments}
%\\
%\nils{Nils comments}
%\\
%\thilo{Thilo comments}
%\\ \\

In this paper we deal with the problem of hedging general contingent claims under illiquidity. Stochastic liquidity cost is incurred by hedging with multiple assets with possibly different levels of liquidity. Our main motivation comes from energy markets. Consider for example an agent hedging an Asian-style call option written on the average spot price $S = (S_{u})_{0 \leq u \leq T_2}$ of energy. Such an option has the payoff
\begin{equation}\label{asian:payoff}
 \left( \frac{1}{T_{2} - T_{1}} \int_{T_{1}}^{T_2} S_{u} du - K \right)^{+}
\end{equation}
for a strike $K$, with a so-called \emph{delivery period} $[T_{1}, T_2]$. The instruments available for hedging such options are futures delivering over the same or a different time period. Hedging is a challenge though since these futures are either not trading in their delivery period at all \citep[such a setting was considered in][]{BenthDetering} or are very illiquid such that hedging incurs large transaction costs. In addition, futures are usually very illiquid for $t\ll T_1$, so their liquidity has a delicate time-structure. In the market, multiple futures with different delivery periods (week, month, quarter, year) and different levels of liquidity are available as hedging instruments. The results of our paper can be applied to hedging options in energy markets with multiple futures by accounting for their different levels of liquidity. The Asian-style option is a particular example but other payoffs as for example Quanto options \cite[see][]{benthQuanto} can be handle
 d equally.
\par
The effect of illiquidity on hedging and optimal trading is a very active research topic in mathematical finance. Still, there is neither an agreed notion of liquidity risk, which is roughly speaking the additional risk due to timing and size of a trade, nor a standard approach for hedging under liquidity costs. A good overview on existing liquidity models in continuous and discrete time can be found in \cite{goekay.roch.soner:2010}.
\par
There are basically two different approaches how to model liquidity risk. The first one is a class of models incorporating feedback effects, that is when the trade volume has a lasting impact on the asset price \citep[see e.g.][]{bank.baum:2004}, also known as permanent price impact or lasting impact. The second approach considers smaller agents whose transactions have no lasting impact on the price of the underlying \cite[see e.g.][and the references therein]{cetin.jarrow.protter:2004}. 
\par
In this paper we stay within the small agents approach and understand liquidity costs as the transaction costs incurred from the hedging strategy by trading through a fast recovering limit order book. In particular we follow the arbitrage-free model suggested by \cite{cetin.jarrow.protter:2004}, who introduced the so called \textit{supply curve} model. There, the asset price is a function of the trade size and the authors developed an extended arbitrage pricing theory.
\par
In addition to the vast majority of papers on illiquid markets dealing with optimal execution, there are also many papers investigating hedging under illiquidity, most of them consider super-replication \cite[see for example][]{bank.baum:2004, cetin.soner.touzi:2010, goekay.soner:2012}. As super-replication is often too expensive we use a quadratic risk criterion. In the classical frictionless theory without transaction costs, there are two main approaches for quadratic hedging \cite[see][for a survey]{schweizer:2001}. First, the \emph{mean-variance} approach, which was introduced in \cite{foellmer.sondermann:1986}, relies on self-financing strategies which produce as final outcome the portfolio $V_{T} := c + \int_{0}^{T} X_{u} dS_{u}$ for some initial capital $c \in \mathbb{R}$ and a trading strategy $X$ in the risky asset $S = (S_{u})_{0 \leq u \leq T}$. The goal of this method is to look for the best approximation of a contingent claim $H$ by the terminal portfolio value $
 V_{T}$, that is minimizing the quadratic hedge error
\begin{equation}
  \E \left[ \left( H - \left(  c + \int_{0}^{T} X_{u} dS_{u} \right) \right)^{2} \right]
\end{equation}
under the real world probability measure and over an appropriately constrained set of strategies. This is also called global risk-minimization. In discrete time, this problem was solved in \cite{schweizer:1995} in a general setting and relaxing the assumptions imposed earlier in \cite{schael:1994}. Later on, this was extended to the multidimensional case with proportional transaction cost in \cite{motoczynski:2000} and \cite{beutner:2007} where the authors show existence of an optimal strategy. The papers \cite{rogers.singh:2010}, \cite{agliardi.gencay:2014} and \cite{bank.soner.voss:2017} can be seen as an extension under illiquidity of the mean-variance hedging criterion, which is based on minimizing the global risk against random fluctuations of the stock price incurring low liquidity costs.
\par
A second quadratic method for hedging in an incomplete market is \emph{local risk-minimization} first introduced in \cite{schweizer:1988} and later extended in \cite{lamberton.pham.schweizer:1998} by accounting for proportional transaction costs in the discrete time case. For discrete time $k=0, \dots, T$ this method does not insist on the self-financing condition but instead the goal is to find a strategy $(X,Y)=(X_{k+1}, Y_{k})_{k=0, \dots, T}$ with book value $V_{k} = X_{k+1}S_{k} + Y_{k}$ (the risk-free asset is assumed constantly equal to $1$) such that  $H = V_{T}$, the cost process $C_{k} = V_{k} - \sum_{m=1}^{k} X_{m} (S_{m} - S_{m-1})$ is a martingale and the variance of the incremental cost is minimized. Here, the strategy $X_{k+1}$ represents the number of shares held in the risky asset and $Y_{k}$ the units held in the bank account in the time interval $(k, k+1]$. In our paper, we extend the work of \cite{schweizer:1988} by considering a multidimensional asset pri
 ce process in discrete time. Secondly we extend the local risk-minimizing quadratic criterion to an illiquid market in the spirit of \cite{rogers.singh:2010} and \cite{agliardi.gencay:2014}. 
\par
In contrast to the existing literature our approach and setting is designed to address the above mentioned problem in energy markets. For this we need a multi-dimensional setup to allow for hedging with multiple futures. Second, the assets price dynamics has to be general enough to capture the characteristics of energy markets and we need a time dependent liquidity structure. Our risk criterion is chosen such that it allows for more explicit formulas for the optimal strategy than existing approaches. Furthermore, as shown in a case study, they are also  computationally tractable. Our main result is the existence of a locally risk-minimizing strategy under illiquidity requiring only mild conditions on the asset price. These conditions are quite technical but they can be reduced to conditions on the covariance matrix of the price process, which can be checked easily for most processes relevant in practice. Furthermore, the strategies can be calculated backwards in time by using
  a \emph{least-squares Monte Carlo} algorithm. %These conditions arise due to the liquidity costs in our setting which can be thought as quadratic transaction costs. 
\par
%Our setup allows us to explore the tradeoff between liquidity and hedging performance in an illiquid market. Hedging for example an Asian-style option of the form (\ref{asian:payoff}) using as hedging instruments a long future and a short one, where the long is more correlated with the option but less liquid than the short one, which is less correlated with the option than the long future, then it is not clear if the hedger needs to use more the shorter or the longer future for hedging the option. 
\par
Our setup allows us to explore the tradeoff between liquidity and hedge quality of available hedge instruments. For example, consider the Asian-style option (\ref{asian:payoff}) in a market where different futures with different delivery periods and different liquidity levels are available for hedging. In such a situation, there are futures with delivery period well matching the delivery period of the option payout resulting in a strong correlation between the future and the option to hedge. However, in certain time periods these {\em hedge-optimal} futures are very illiquid and futures which are less correlated but more liquid might be better for hedging. Our framework allows us to explore this tradeoff and provide market-makers with a more profound tool for risk management.
\par
The paper is structured as follows. Section \ref{sec:model} explains the model framework and describes the basic problem. In Section \ref{sec:linear:supply} we focus on a linear supply curve and impose necessary assumptions on the price process to prove our main existence Theorem~\ref{theo:existence}. Sufficient conditions to check the assumptions are also provided. Section \ref{sec:applications} considers an application to energy markets. Optimal strategies under illiquidity are simulated by means of a least-squares Monte Carlo method. This allows us to explore the tradeoff between liquidity and hedging performance of futures available for hedging. 
%Finally, Section \ref{sec:conclusions} concludes.

%The Model
\section{The Model}
\label{sec:model}

Given a filtered probability space $(\Omega, \mathcal{F}, \mathbb{F}, \mathbb{P})$ consider a financial market consisting of $d+1$ assets.  
  We denote by $\mathbb{P}$ the {\em objective} probability measure and by $\mathbb{F}$ the filtration $\mathbb{F} = (\mathcal{F}_k)_{k = 0,1, \dots, T} $, which describes the flow of information. We shall use the indices $k = 0,1, \dots, T$ to refer to a discrete time grid with time points $t_{0} < t_{1} < \dots < t_{T}$ and sometimes use both interchangeably. An $\mathbb{F}$-adapted, nonnegative $d$-dimensional stochastic process $S = (S_{k})_{k=0,1, \dots, T}$ describes the discounted price of $d$ risky assets (typically futures or stocks). We use $S_{k}^{j}$ to refer to the price of asset $j$ at time $t_{k}$.  Furthermore, a riskless asset (typically a bond) exists whose discounted price is constantly $1$.
  \par
  Similar as in  \cite{cetin.jarrow.protter:2004} we assume that a hedger observes an exogenously given nonnegative $d$-dimensional \textit{supply curve} $S_{k}(x)$ where $S_{k}(x)^{j} := S_{k}^{j}(x^{j})$ represents the $j$-th stock price per share at time $k$ for the purchase (if $x^{j} > 0$) or sale (if $x^{j} < 0$) of $|x^{j}|$ shares. We call $S(0) = S$ the \textit{marginal price}. 
  The supply curve determines the actual price that market participants pay or receive respectively for a transaction of size $x$ at time $k$. This curve is also assumed to be independent of the participants past actions which implies no lasting impact of the trading strategy on the supply curve. The only assumption that we need for the moment is measurability of the supply curve w.r.t. the filtration $\mathbb{F}$ and that it is non-decreasing in the number of shares $x$, i.e. for each $k$ and $j$, $S_{k}(x)^{j} \leq S_{k}(y)^{j}, \, \mathbb{P}-a.s.$ for $x^{j} \leq y^{j}$. This will ensure that the liquidity costs are non-negative.
\par
In \cite{cetin.jarrow.protter:2004} the authors develop a continuous time version of such a supply curve model and an extended arbitrage pricing theory. They show that the existence of an equivalent local martingale measure $\mathbb{Q}$ for the marginal price process $S$ rules out arbitrage. A similar results can easily be seen to hold in our setting as liquidity cost is always positive. 
\par
 However, even a unique martingale measure (and state space restrictions in a discrete setting) do not necessarily ensure completeness if one incorporates illiquidity. Since we cannot hedge perfectly, we want to minimize locally the risk of hedging under illiquidity according to an optimality criterion introduced in Definition \ref{defi:local risk minimizing strategy under illiquidity}.
\par
In the following we shall consider an investor who aims at hedging an $\mathcal{F}_T$-measurable claim $H$. For $x\in \mathbb{R}^d$, let $\abs{x}$ denote the Euclidian norm and $x^{*}$ the transpose of $x$. Further, $\langle x , y \rangle$ denotes the inner product of $x,y\in \mathbb{R}^d$. Adapting \cite{schweizer:1988} we define the investor's possible trading strategies. For this we denote by $\mathbb{L}_{T}^{p} (\mathbb{R}^d)$ (in short $\mathbb{L}_{T}^{p, d}$), 
    the space of all $\mathcal{F}_T$-measurable
    random variables $Z : \Omega \to \mathbb{R}^d$ satisfying 
    $\norm{Z}^{p} = \mathbb{E}( \abs{Z}^{p} ) < \infty$. We abbreviate $\Delta S_{k} := S_{k} - S_{k-1}$. Furthermore, we denote by $\Theta_{d}(S)$ the space of all $\mathbb{R}^{d}$-valued predictable strategies $X = (X_{k})_{k=1,2, \dots, T+1}$ so that $ X_{k}^{*} \Delta S_{k} \in \mathbb{L}_{T}^{2, 1}$ and $\Delta X_{k+1}^{*} [ S_{k}(\Delta X_{k+1}) - S_{k}(0) ] \in \mathbb{L}_{T}^{1, 1}$ for $k=1, 2, \dots, T$.% and by $\hat \Theta_{d}(S)$ the space of all $\mathbb{R}^{d}$-valued predictable strategies $X = (X_{k})_{k=1,2, \dots, T+1}$ so that $ X_{k}^{*} \Delta S_{k} \in \mathbb{L}_{T}^{2, 1}$ for $k=1, 2, \dots, T$. 
  
  \begin{defi}
    \label{defi:trading strategy}
    A pair $\varphi = (X,Y) $ is called a \textit{trading strategy} if:
    \begin{enumerate}[label=(\roman*)]
%    \item $X = (X_{k})_{k=1,2, \dots, T+1}$ is a $d$-dimensional real-valued $\mathbb{F}$-predictable process. 
    \item $Y = (Y_{k})_{k=0,1, \dots, T}$ is a real-valued $\mathbb{F}$-adapted process. 
    \item \label{item:trading strategy L2int}
    $X\in \Theta_{d}(S)$.
    %$ X_{k}^{*} \Delta S_{k} \in \mathbb{L}_{T}^{2, d}$ for $k=1, 2, \dots, T$.
    \item $V_{k}(\varphi) := X_{k+1}^{*} S_{k} + Y_{k} \in \mathbb{L}_{T}^{2, 1}$ for $k=0, 1, \dots, T$.
    \end{enumerate}
  \end{defi}
%\nils{deleted (Definition \ref{defi:trading strategy} sagt nur was ueber Paare $\varphi = (X,Y) $ aus: Definition \ref{defi:trading strategy} \ref{item:trading strategy L2int} holds.}
    \par
    We call $V_{k}(\varphi)$ the \textit{marked-to-market value} or the \textit{book value} of the portfolio $(X_{k+1},Y_{k})$ at time $k$. We interpret $X_{k+1}^{j}$ as the number of shares held in the risky asset $S_{k}^{j}$ and $Y_{k}$ the units held in the riskless asset (bank account) in the time interval $( k, k+1 ]$. Note that with a non-flat supply curve, there is no unique value for a portfolio, as the value that can be realized depends on the liquidation strategy.

    % Cost and Risk process
    \subsection{Cost And Risk process}
    \label{sec:cost and risk process}

   Consider an $\mathbb{L}_{T}^{2, 1}$-\textit{contingent claim} of the form $H = \bar X_{T+1}^{*} S_{T} + \bar Y_{T}$, where $ \bar X_{T+1}^{*}S_{T} \in \mathbb{L}_{T}^{2, 1}$, $ \bar X_{T+1} \in \mathbb{L}_{T}^{2, d}$ and the components of the pair $(\bar X_{T+1}, \bar Y_{T})$ are $\mathcal{F}_{T}$-measurable random variables describing the quantity in risky assets and bonds respectively that the option seller is committed to provide to the buyer at the expiration date $T$ of the financial contract $H$.\footnote{For example, in the $1$-dimensional case one could set $\bar X_{T+1} = 0$ and $Y_{T}=(S_{T} - K)^{+}$ for a call option with strike $K$ without physical delivery.}
   \par
   Assuming that at time $k \in \{ 1, 2, \dots, T \}$ an order of $\Delta Y_{k}:= Y_{k} - Y_{k-1}$ bonds and $\Delta X_{k+1}:=X_{k+1} - X_{k}$ shares is made, then the \textit{total outlay} (under liquidity costs) is
   \begin{equation}
     \label{eq:total outlay}
       \Delta Y_{k} + \Delta X_{k+1}^{*} S_{k}(\Delta X_{k+1})
         = \Delta Y_{k} + \Delta X_{k+1}^{*} S_{k} + \Delta X_{k+1}^{*} [ S_{k}(\Delta X_{k+1}) - S_{k}(0) ].
   \end{equation}
Note that $S_{k}(0)= S_{k}$ is the marginal price, such that the last term can be seen as the transaction cost resulting from market illiquidity. Furthermore, using the definition of the book value the previous equation can be written as
   \begin{equation}
     \label{eq:total outlay 2}
       \Delta Y_{k} + \Delta X_{k+1}^{*} S_{k}(\Delta X_{k+1})
         = \Delta V_{k} (\varphi) - X_{k}^{*} \Delta S_{k} + \Delta X_{k+1}^{*} [ S_{k}(\Delta X_{k+1}) - S_{k}(0) ] 
   \end{equation}
   For a self-financing trading strategy the total outlay at time $k$ would be zero.
   \\ 
   
   Now, by defining $\hat C_{0}(\varphi) := V_{0}(\varphi)$, the \textit{initial cost}\footnote{For simplicity  we do not account for any liquidity costs paid to set up the initial portfolio.}, we can define the \textit{cost process under illiquidity} $\hat C(\varphi) = (\hat C_{k}(\varphi))_{k=0,1, \dots, T}$ as 
   \begin{equation}
       \hat C_{k}(\varphi) 
         := \sum_{m=1}^{k} \Delta Y_{m} + \sum_{m=1}^{k} \Delta X_{m+1}^{*} S_{m}(\Delta X_{m+1}) + V_{0}(\varphi).
   \end{equation}
   It quantifies the cumulative costs of the strategy $\varphi = (X, Y)$. A simple calculation using the definition of $V_{k} (\varphi)$ shows that 
\begin{equation} 
\label{eq:cost process} 
\hat C_{k}(\varphi) =V_{k} (\varphi) - \sum_{m=1}^{k} X_{m}^{*} \Delta S_{m} + \sum_{m=1}^{k} \Delta X_{m+1}^{*} [ S_{m}(\Delta X_{m+1}) - S_{m}(0) ],
\end{equation}
which will be needed later. If we can ensure that the cost process is square integrable, then we can define the \textit{quadratic risk process under illiquidity} $\hat R (\varphi)= (\hat R_{k}(\varphi))_{k=0,1, \dots, T}$ by
   \begin{equation}
     \label{eq:quadratic risk process}
       \hat R_{k}(\varphi)
         := \condE{ (\hat C_{T}(\varphi) - \hat C_{k}(\varphi) )^{2} }   \,.
   \end{equation}
   At this point we would like to mention that the classical local risk-minimization approach aims at finding a locally risk-minimizing strategy $\varphi = (X, Y)$ such that $V_{T}(\varphi) = H$ with $X_{T+1} = \bar X_{T+1}$ and $Y_{T} = \bar Y_{T}$ (see Section \ref{sec:description of the basic problem} for more details).
   \par
   We denote by $C(\varphi) = (C_{k}(\varphi))_{k=0,1, \dots, T}$ the classical cost process without liquidity costs (i.e., $S(x)=S(0)$), that is
      \begin{equation}
 C_{k}(\varphi) 
         := V_{k} (\varphi) - \sum_{m=1}^{k} X_{m}^{*} \Delta S_{m},
      \end{equation}
   and obtain the relation 
   \begin{equation}
     \hat C_{T}(\varphi) - \hat C_{k}(\varphi)
       = C_{T}(\varphi) - C_{k}(\varphi) + \sum_{m=k+1}^{T} \Delta X_{m+1}^{*} [ S_{m}(\Delta X_{m+1}) - S_{m}(0) ].
   \end{equation}
   Furthermore, we denote by $ R (\varphi)= (R_{k}(\varphi))_{k=0,1, \dots, T}$ the classical risk process, defined as in (\ref{eq:quadratic risk process}) but with $\hat C$ replaced by $C$.
   
   One could also define the \textit{linear risk process under illiquidity} 
     \begin{equation}
     \label{eq:linear risk process}
       \bar R_{k}(\varphi)
         := \condE{ | \hat C_{T}(\varphi) - \hat C_{k}(\varphi) | }
   \end{equation}
   which is motivated in \cite{coleman.li.patron:2003}. 
    \begin{remark}
   A linear local risk-minimization criterion might be more natural than a quadratic one from a financial perspective. The $L^2$-norm overemphasizes large values even if these values occur with small probability. Nevertheless, by minimizing over the $L^2$-norm, it is possible to get explicit results.
   \end{remark}
   A combination of the two, that means measuring the quadratic difference of the classical cost process and linearly the liquidity costs, yields the \textit{quadratic-linear risk process under illiquidity}.
   \begin{equation}
     \label{eq:quadratic linear risk process}
       T_{k}(\varphi)
         := \condE{ ( C_{T}(\varphi) - C_{k}(\varphi) )^{2} }
             + \condE{ \sum_{m=k+1}^{T} \Delta X_{m+1}^{*} [ S_{m}(\Delta X_{m+1}) - S_{m}(0) ] }.
   \end{equation}
   As we will see later on, by minimizing the expression in (\ref{eq:quadratic linear risk process}) we will be able to construct an explicit representation of the LRM-strategy under illiquidity where large values of liquidity costs are not overemphasized by the $L^2$-norm.
   %Description of the Basic problem
   \subsection{Description of the basic problem}
   \label{sec:description of the basic problem}

   The aim of the classical local risk-minimization is to minimize locally the conditional mean square incremental cost of a strategy. Our criterion is targeting on minimizing locally the risk against random fluctuations of the stock price but at the same time reducing liquidity costs. It balances low liquidity costs against poor replication. Such an approach is similar to \cite{agliardi.gencay:2014} or \cite{rogers.singh:2010} and yields a tractable problem. 
   
In minimizing locally the risk process at time $k$, we only minimize in $Y_{k}$ and $X_{k+1}$ in order to make the current optimal choice of the strategy by fixing the holdings at past or future times. Definition \ref{defi:local perturbation} and Definition \ref{defi:local risk minimizing strategy under illiquidity} give us the optimality criterion that the minimization problem is based on.
   \begin{defi}
     \label{defi:local perturbation}
       A \textit{local perturbation} $ \varphi ' = ( X', Y')$ of a strategy $\varphi = (X, Y)$ at time $k \in \{ 0,1, \dots, T-1 \}$ is a trading strategy such that $X_{m+1} = X'_{m+1}$ and $Y_{m} = Y'_{m}$ for all $m \neq k$. 
   \end{defi}
   By a slight abuse of notation let 
   \begin{equation}
     \label{eq:quadratic linear risk process new}
       T_{k}^{\alpha}(\varphi)
         := \condE{ ( C_{T}(\varphi) - C_{k}(\varphi) )^{2} }
             + \alpha \condE{ \Delta X_{k+2}^{*} [ S_{k+1}(\Delta X_{k+2}) - S_{k+1}(0) ] }.
   \end{equation}
   We specify in Definition \ref{defi:local risk minimizing strategy under illiquidity} what we call local risk minimizing (LRM) strategy under illiquidity for some $\alpha \in \mathbb{R^{+}}$.
   \begin{defi}
     \label{defi:local risk minimizing strategy under illiquidity}
       A trading strategy $\varphi = (X, Y)$ is called \textit{locally risk-minimizing under illiquidity} if
       \begin{equation}
         T_{k}^{\alpha} (\varphi)
           \leq T_{k}^{\alpha} (\varphi ' )
           \quad \quad \Pas
       \end{equation}
       for any time $k \in \{ 0,1, \dots, T-1 \}$ and any local perturbation $\varphi '$ of $\varphi$ at time $k$.
   \end{defi}
     Definition \ref{defi:local risk minimizing strategy under illiquidity} assumes that for any strategy the classical cost process $C$ is square-integrable and the liquidity costs are integrable. By Definition \ref{defi:trading strategy} this is ensured. Note also that in Definition \ref{defi:local risk minimizing strategy under illiquidity} we have only taken into account the liquidity costs at the current time. This is equivalent to minimizing over $T_{k}$ in equation (\ref{eq:quadratic linear risk process}) since we minimize only locally.
   \begin{remark}
   The choice $\alpha = 1$ represents an equal concern about  the risk to be hedged as incurred by market price fluctuations and the cost of hedging incurred by liquidity costs. Otherwise, $\alpha < 1$ means a major risk aversion to the risk of miss-hedging and $\alpha > 1$ a major risk aversion to the cost of illiquidity. One could also generalize by having a deterministic $\mathbb{R}$-valued process $\alpha = (\alpha_{k})_{k=0, 1, \dots, T}$ and trivially our results will still hold.
   \end{remark}     
So in the following we assume $\alpha$  is given and we aim at finding a locally risk-minimizing strategy $\varphi = (X, Y)$ under illiquidity such that $V_{T}(\varphi) = H$ with $X_{T+1} = \bar X_{T+1}$ and $Y_{T} = \bar Y_{T}$.
   Some useful Lemmas follow, which even in the multi-dimensional case, can be shown by means very similar to those used in \cite{lamberton.pham.schweizer:1998}. For completeness we provide their proofs in Appendix~\ref{appendix}.
   
The first Lemma shows that a main property of a local risk-minimizing strategy, namely that its cost process is a martingale, generalizes to our setting. The reason is that a strategy $\varphi$ can be perturbed to $\varphi '$ such that $C(\varphi ')$ is a martingale by changing only the $\mathcal{F}_{k}$-measurable risk free investment. This in turn reduces the first term in (\ref{eq:quadratic linear risk process new}) but leaves the second term unchanged.
   
   \begin{lem}
     \label{lem:cost process martingale}
     For a LRM-strategy $\varphi$ under illiquidity, the cost process $C(\varphi)$ is a martingale. Furthermore, from the martingale property of the cost process we get the representation, 
     \begin{equation}
     \label{eq:cost process martingale}
       R_{k}(\varphi)
         = \condE{ R_{k+1}(\varphi) }
           + \condVar{ \Delta C_{k+1}(\varphi) }
           \quad \quad \Pas
     \end{equation}
     for $k = 0,1, \dots, T-1$.
   \end{lem} 
   So, for $\varphi$ a LRM-strategy under illiquidity, the quadratic-linear risk process (QLRP) under illiquidity has the representation
   \begin{equation}
   \label{eq:quadratic linear risk process modified}
     T_{k}^{\alpha}(\varphi)
         = \condE{ R_{k+1}(\varphi) }
             + \condVar{ \Delta C_{k+1}(\varphi) }
             + \alpha \condE{ \Delta X_{k+2}^{*} [ S_{k+1}(\Delta X_{k+2}) - S_{k+1}(0) ] }.
   \end{equation}
   \par
   The next lemma provides a representation for the QLRP process of a perturbed strategy.
   
   \begin{lem}
   \label{lem:quadratic linear risk process modified}
     If $C(\varphi)$ is a martingale and $\varphi '$ a local perturbation of $\varphi$ at time $k$ then 
     \begin{align}
     \label{eq:quadratic linear risk process modified}
       T_{k}^{\alpha}(\varphi ')
         = & \condE{ R_{k+1}(\varphi) }
             + \condE{ (\Delta C_{k+1}(\varphi '))^{2} }       \nonumber 
            \\ & + \alpha \condE{ (X_{k+2} - X'_{k+1})^{*} [ S_{k+1}(X_{k+2} - X'_{k+1}) - S_{k+1}(0) ] }   \,.
     \end{align}
   \end{lem}
   
   \begin{remark}
     Since $R_{k+1}(\varphi)=R_{k+1}(\varphi ')$ for any local perturbation $\varphi '$ of $\varphi$ at time $k$, it follows from equation (\ref{eq:quadratic linear risk process modified}) that one needs to minimize over
     \begin{equation}
       \condVar{ \Delta C_{k+1}(\varphi) }
             + \alpha \condE{ \Delta X_{k+2}^{*} [ S_{k+1}(\Delta X_{k+2}) - S_{k+1}(0) ] }
     \end{equation}
     at time $k$ (see also Proposition \ref{prop:minimizing variance}).
   \end{remark}
   \begin{prop}
     \label{prop:minimizing variance}
     A trading strategy $\varphi = (X, Y)$ is LRM under illiquidity if and only if the two following properties are satisfied:
     \begin{enumerate}[label=(\roman*)]
     \item \label{prop:1}$C(\varphi)$ is a martingale.
     \item \label{prop:2}For each $k \in \{ 0,1, \dots, T-1 \}, X_{k+1}$ minimizes
     \begin{align}
       & \condVar{ V_{k+1}(\varphi) - (X'_{k+1})^{*} \Delta S_{k+1} }     \nonumber 
       \\& \quad \quad + \alpha \condE{ (X_{k+2} - X'_{k+1})^{*} [ S_{k+1}(X_{k+2} - X'_{k+1}) - S_{k+1}(0) ] }
     \end{align}
     over all $\mathcal{F}_{k}$-measurable random variables $X'_{k+1}$ so that $(X'_{k+1})^{*} \Delta S_{k+1} \in \mathbb{L}_{T}^{2,1}$ and $(X_{k+2} - X'_{k+1})^{*} [ S_{k+1}(X_{k+2} - X'_{k+1}) - S_{k+1}(0) ]  \in \mathbb{L}_{T}^{1,1}$.
     \end{enumerate}
   \end{prop}

Proposition \ref{prop:minimizing variance} is quite general since it holds for any supply curve. For the existence and recursive construction of a LRM-strategy under illiquidity we will consider in the next section a special case of the supply curve that is motivated from a multiplicative limit order book. For this model we can construct explicitly the optimal strategy and we are able to state conditions that ensure that the optimal strategy belongs to the space $\Theta_{d}(S)$.
   
  \section{Linear supply curve}\label{sec:linear:supply}
  %   \subsection{Existence and recursive construction of an optimal strategy under illiquidity.}
%   \label{sec:existence and recursive construction of an optimal strategy under illiquidity}
   %Existence and recursive construction of an optimal strategy under illiquidity   
   
   When trading through a limit order book (LOB) in an illiquid environment, liquidity costs are related to the depth of the order book. We do not take into account any feedback effects from hedging strategies, so we assume that the speed of resilience, i.e., the ability of the order book to recover itself after a trade, is infinite. 
    We choose the form of the supply curve $S_{k}(x) = (S_{k}^{1}(x^{1}), \dots, S_{k}^{d}(x^{d}))^{*}$ to be
   \begin{equation}
     S_{k}^{j}(x^{j}) = S_{k}^{j} + x^{j} \varepsilon_{k}^{j} S_{k}^{j}
   \end{equation}
   and assume that the price process $S$ is a non-negative semimartingale and $\varepsilon = ( \varepsilon_{k} )_{k = 0,1, \dots, T} $ is a positive deterministic $\mathbb{R}^{d}$-valued process. Note that it is possible for $S_{k}(x)$ to take negative values for some $x$, but in practice this is unlikely to happen for small values of $\varepsilon_k$ and reasonable values of $x$.
   \par
   Now let us describe a (multiplicative) limit order book for the specific form of the supply curve. A symmetric, $1$-dimensional, time independent (for simplicity) LOB is represented by a density function $q$, where $q(x)dx$ is the bid or ask offers at price level $xS_{k}$. Denote by $F(\rho) = \int_{1}^{\rho} q(x) dx$ the quantity to buy up through the LOB at price $\rho S_{k}$. If an investor makes an order of $x = F(\rho)$ shares through the LOB at time $k$ then some limit orders are eaten up and the quoted price is shifted up to $S_{k}(x)^{+} := g(x)S_{k}$ where $g(x)$ solves the equation $x = \int_{1}^{g(x)} q(y) dy$, that is $g(x) = F^{-1}(x)$.\footnote{Note the multiplicative way of shifting up the price. In an additive LOB this would be of the form $S_{k}(x)^{+} := S_{k} + g(x)$ as for example in \cite{roch:2011}. For a description of multiplicative and additive limit order books see for example \cite{lokka:2012}.} Since we do not account for any price impact, then 
 after the trade, the price returns to $S_{k}$.\footnote{In the literature, this is the so-called resilience effect, measuring the proportion of new bid or ask orders filling up the LOB after a trade. In our case we have infinite resilience.} The cost of an order of $x$ shares will be $S_{k} \int_{1}^{g(x)} \rho dF(\rho)$ which for an appropriate choice of the function $q$, should be equal to $xS_{k}(x) = xS_{k} + \varepsilon x^{2} S_{k}$. Choosing an, independent from price, density
   \begin{equation}
     q(x) = \frac{1}{2 \varepsilon}
   \end{equation}
does the job. The process $\varepsilon$ can be thought as a measure of illiquidity. For $\varepsilon$ tending to zero the market becomes more liquid and the liquidity cost vanishes. 
   \begin{remark}
     Recall that the supply curve $S_{k}(x)$ is increasing in the transaction size $x$ which ensures non-negative liquidity cost, that is $x[S_{k}(x) - S_{k}(0)] \geq 0$. The specific choice of a linear supply curve implies $\varepsilon S_{k} | x |^{2}$ liquidity costs for a transaction of size $x$ at time $k$. Note that it is essential to assume that the marginal price process $S$ is non-negative in order to avoid negative liquidity costs.
          Note that when the price process $S_{k}$ increases, then naturally also the liquidity cost increases but not the availability of assets in the LOB since the depth of the order book $q_{k}(y) = \frac{1}{2 \varepsilon_{k}}$ depends only on $\varepsilon_{k}$.  
   \end{remark}
   Proposition \ref{prop:minimizing variance} tells us how to construct an optimal strategy according to the LRM-criterion under illiquidity. Going backward in time we need to minimize at time $k$
    \begin{align}
      \label{eq:minimizing expression}
       & \condVar{ V_{k+1}(\varphi) - (X'_{k+1})^{*} \Delta S_{k+1} }           \nonumber
       \\& \quad \quad \quad \quad 
       + \alpha \condE{ \sum_{j=1}^{d} \varepsilon_{k+1}^{j} S_{k+1}^{j} (X_{k+2}^{j} - (X'_{k+1})^{j})^{2} }         
     \end{align}
    over all appropriate $X'_{k+1}$ (see Definition \ref{defi:trading strategy}) and choose $Y_{k}$ so that the cost process $C$ becomes a martingale. 
    
    Before continuing let us first introduce some notation: 
    \begin{align}
      A_{k; j}^{0} &:= \condVar{ \Delta S_{k+1}^{j} } 
      \quad & \quad
      A_{k; j}^{\varepsilon} &:= \condE{ \varepsilon_{k+1}^{j} S_{k+1}^{j} }
      \quad & \quad
      A_{k; j} &:= A_{k; j}^{0} + A_{k; j}^{\varepsilon}                        \nonumber
      \\
      b_{k; j}^{0} &:= \condCov{ V_{k+1}, \Delta S_{k+1}^{j} }
      \quad & \quad
      b_{k; j}^{\varepsilon} &:= \condE{ \varepsilon_{k+1}^{j} S_{k+1}^{j} X_{k+2}^{j} }
      \quad & \quad
      b_{k; j} &:= b_{k; j}^{0} + b_{k; j}^{\varepsilon}                         \nonumber
      \\
      D_{k; j, i} &:= \condCov{ \Delta S_{k+1}^{j}, \Delta S_{k+1}^{i} }
      & \text{ for } i \neq j &
    \end{align}
    for all $i,j = 1, \dots, d$ and $k=0,\dots,T-1$.
    \par
    Furthermore, we can rewrite the expression (\ref{eq:minimizing expression}) by defining the function $f_{k} : \mathbb{R}^{d} \times \Omega \rightarrow \mathbb{R}^{+}$ as 
    \begin{align}
      \label{eq:minimizing function}
      f_{k}(c, \omega)
        &= \sum_{j=1}^{d} | c_{j} |^{2} A_{k; j}(\omega)
          - 2 \sum_{j=1}^{d} c_{j} b_{k; j}(\omega)
          + \sum_{j \neq i}  c_{j} c_{i} D_{k; j, i}(\omega)                        \nonumber 
        \\
        &+ \condVar{ V_{k+1} }(\omega) 
          + \sum_{j=1}^{d} \condE{ 
               \varepsilon_{k+1}^{j} S_{k+1}^{j} | X_{k+2}^{j} |^{2} }(\omega). 
    \end{align}
    Fixing $\omega$ one can easily calculate the gradient of the multidimensional function $f_{k}$. We need to solve $grad(f_{k}) = 0$ to calculate the candidates of extreme points which translates into solving a linear equation system of the form
   \begin{equation}\label{c:equation}
     F_{k} \, c = b_{k}
   \end{equation}
   where $F_{k} \in \mathbb{R}^{d \times d}$ with $F_{k; i, j} = D_{k; i, j}$ for $i \neq j$, $F_{k; i, j} = A_{k; j}$ for $i = j$ and $b_{k} = ( b_{k;1}, \dots, b_{k; d} ) \in \mathbb{R}^{d}$. Let $F_{k}^{\varepsilon} = diag(A_{k; 1}^{\varepsilon}, \dots, A_{k; d}^{\varepsilon})$ and denote by $F_{k}^{0}$ the matrix $F_{k}$ with $\varepsilon_{k+1}^{j} = 0$ for all $j$, that is the covariance matrix of the marginal price process $S$. Then the symmetric matrix $F_{k}$ is the sum of two real symmetric, positive semidefinite matrices $F_{k} = F_{k}^{0} + F_{k}^{\varepsilon}$. This implies that the matrix $F_{k}$ is also positive semidefinite\footnote{In fact, $F_{k}$ is positive definite if $\varepsilon_{k+1}^{j}$ is positive for all $j=1, \dots, d$.} and therefore also the Hessian matrix which calculates as $H_{f_{k}}(c) = 2 F_{k}$. So, assuming that the covariance matrix $F_{k}^{0}$ is positive definite, this implies that $F_{k}$ is invertible and equation (\ref{c:equation}) h
 as a unique solution. Furthermore, since also the Hesse matrix is positive definite the function $c \rightarrow f_{k}(c, \omega)$ is strictly convex, which implies that $c^{*} := F_{k}^{-1} b_{k}$ is a global minimizer. Furthermore, since the matrix $F_{k}^{-1}$ and $b_{k}$ are both $\mathcal{F}_{k}$-measurable it is clear that the minimizer $c^{*}$ is also $\mathcal{F}_{k}$-measurable.
   
   %Assumptions
   \subsection{Properties of the marginal price process $S$}
   \label{sec:assumptions}
 In order to show that the optimal strategy $c^{*}$ calculated above belongs to the space $\Theta_{d}(S)$, we need slightly stronger assumptions on the matrix $F_{k}$, which can be reduced to assumptions on the covariance matrix of $S$. We will impose these assumptions now. It will turn out that they hold for independent increments as well as for independent returns. 
%\newver{(In this section we replace all constants with small letters e.g. $c, \tilde c$ etc. with constants with big letters e.g. $C, \tilde C$ etc.)}
 
   \begin{defi}
   \label{defi:bounded mean-variance tradeoff}
     We say that $S$ has \textit{bounded mean-variance tradeoff} process if for some constant $C > 0$ 
     \begin{equation}
       \frac{( \condE{ \Delta S_{k+1}^{j} } )^{2} }{ \condVar{ \Delta S_{k+1}^{j} } }
       \leq C
       \quad \Pas \quad \text{for all } j = 1, \dots, d 
     \end{equation}   
     uniformly in $k$ and $\omega$.
   \end{defi}
   \begin{defi}
     \label{defi:modified bounded mean-variance tradeoff}
     We say that $S$ has \textit{modified above bounded mean-variance tradeoff} process if for some constant $C > 0$ 
     \begin{equation}
       \frac{( \condE{ S_{k+1}^{j} } )^{2} }{ \condVar{ S_{k+1}^{j} } }
       \leq C
       \quad \Pas \quad \text{for all } j = 1, \dots, d 
     \end{equation}   
     uniformly in $k$ and $\omega$. Furthermore $S$ has \textit{modified below bounded mean-variance tradeoff} process if for some constant $\tilde C > 0$ 
     \begin{equation}
       \frac{( \condE{ S_{k+1}^{j} } )^{2} }{ \condVar{ S_{k+1}^{j} } }
       \geq \tilde C
       \quad \Pas \quad \text{for all } j = 1, \dots, d 
     \end{equation}   
     uniformly in $k$ and $\omega$. If both bounds hold then we say that $S$ has \textit{modified bounded mean-variance tradeoff}.
   \end{defi}
   \begin{remark}
     Note that for the case of $S$ being a submartingale then the property of modified above bounded mean-variance tradeoff implies that of bounded mean-variance tradeoff, since by using $(a + b)^{2} \leq 2a^{2} + 2b^{2}$ we can estimate
     \begin{align}
       (\condE{ \Delta S_{k+1}^{j} } )^{2}
       \leq 2 (\condE{ S_{k+1}^{j} } )^{2} + 2 | S_{k} |^{2}
       \leq 4 (\condE{ S_{k+1}^{j} } )^{2}
     \end{align}
     where we have also used the fact that $S_{k}^{j}$ is positive.
   \end{remark}
   \begin{defi}
     \label{defi:diagonal condition F}
     We say that $S$ satisfies the \textit{F-diagonal condition} if for some constant $C > 0$
     \begin{equation}
     \label{eq:diagonal condition F ineq 1}
       \sqrt{\condVar{ \Delta S_{k+1}^{j} } }
       + \frac{ \condE{ S_{k+1}^{j} } }{ \sqrt{\condVar{ S_{k+1}^{j} } } }
       \geq C
       \quad \Pas \quad \text{for all } j = 1, \dots, d 
     \end{equation}
     uniformly in $k$ and $\omega$ and if for some constant $\tilde C > 0$
      \begin{equation}
      \label{eq:diagonal condition F ineq 2}
        \frac{ \sqrt{\condVar{ S_{k+1}^{j} } }  }{ \condE{ S_{k+1}^{j} } }
        + \frac{ 1 }{ \sqrt{\condVar{ \Delta S_{k+1}^{j} } } }
       \geq \tilde C
       \quad \Pas \quad \text{for all } j = 1, \dots, d 
     \end{equation}
     uniformly in $k$ and $\omega$.
   \end{defi}
   \begin{remark}
     The name \textit{F-diagonal condition} in Definition \ref{defi:diagonal condition F} comes from the diagonal terms of the matrix $F$, since 
     \begin{align}
       \frac{ F_{k; j, j}^{0} }{ | F_{k; j, j} |^{2} }
       &= \left( \sqrt{\condVar{ \Delta S_{k+1}^{j} } }
       + \frac{ \condE{ \varepsilon_{k+1}^{j} S_{k+1}^{j} } }{ \sqrt{\condVar{ S_{k+1}^{j} } } } \right)^{-2}
       \nonumber  \\
       | F_{k; j, j}^{\varepsilon} |^{2} \frac{ F_{k; j, j}^{0} }{ | F_{k; j, j} |^{2} }
       &= \left( \frac{ \sqrt{\condVar{ S_{k+1}^{j} } }  }{ \condE{ \varepsilon_{k+1}^{j} S_{k+1}^{j} } }
        + \frac{ 1 }{ \sqrt{\condVar{ \Delta S_{k+1}^{j} } } } \right)^{-2}   \,.
     \end{align}
   \end{remark}
   Writing $S_{k+1}^{j} = S_{k}^{j}( 1 + \rho_{k+1}^{j} )$ for $j = 1, \dots, d$ , we denote by $\rho = (\rho_{k})_{k=0,1, \dots T}$ the $d$-dimensional return process of $S$.
   \par
   The next two Propositions \ref{prop:independent increments} and \ref{prop:independent returns} give sufficient conditions for the previous properties on the marginal price process $S$ to hold.
   \begin{prop}
     \label{prop:independent increments}
       For $S$ satisfying $\tilde C \leq \condVar{ \Delta S_{k+1}^{j} } \leq C$ for some positive constants $C, \tilde C$ and for all $j = 1, \dots, d$, then the $F$-diagonal condition holds. In particular, if $S$ has independent increments then $S$ has bounded mean-variance tradeoff and satisfies the $F$-diagonal condition.
   \end{prop}
   \begin{proof}
     The claim follows directly from the fact that $\tilde C \leq \condVar{ \Delta S_{k+1}^{j} } \leq C$.
   \end{proof}
   \begin{prop}
     \label{prop:independent returns}
       For $S$ having modified bounded mean-variance tradeoff then the $F$-diagonal condition holds. In particular, if $S$ has independent returns then $S$ has bounded mean-variance tradeoff and satisfies the $F$-diagonal condition.
   \end{prop}
   \begin{proof}
     The claim follows directly from the fact that $S$ has modified bounded mean-variance tradeoff.
   \end{proof}
   \begin{remark}
     Consider the $1$-dimensional Black-Scholes model of a geometric Brownian motion $W$, that is
     \begin{equation}
       S_{kh} = S_{0} \exp{(bkh + \sigma W_{kh})}
     \end{equation}
     with discretization time step $\Delta t = h$. Then the return process $\rho_{k}$ can be defined by, 
     \begin{equation}
       1 + \rho_{k} = \frac{ S_{kh} }{ S_{(k-1)h} }
     \end{equation}
     and is lognormally distributed. This is also a process of i.i.d. random variables. By Proposition \ref{prop:independent returns}, $S$ has bounded mean-variance tradeoff and satisfies the $F$-diagonal condition.
   \end{remark}
   \par
   
      \subsection{Some preliminaries}
   \label{sec:existence and recursive construction of an optimal strategy under illiquidity}
   Now let us state some useful Lemmas needed in the proof of Theorem \ref{theo:existence} in order to show that the integrability conditions are fulfilled. In what follows we will use the notation
   \begin{align}
     \alpha_{k;i,j} &:= F_{k;  j,  j}^{0}  F_{k; i, i}^{0} | F_{k; j, i}^{-1} |^{2}
     \quad \quad 
     && \alpha_{k;i,j}^{\varepsilon} := F_{k;  j,  j}^{0} | F_{k; i, i}^{\varepsilon} |^{2} | F_{k; j, i}^{-1} |^{2}
     \nonumber  \\
     \beta_{k;i,j} &:= F_{k; i, i}^{0} | F_{k; j, i}^{-1} |^{2}
     \quad \quad
     && \beta_{k;i,j}^{\varepsilon} := | F_{k; i, i}^{\varepsilon} |^{2} | F_{k; j, i}^{-1} |^{2}
   \end{align}  
   for $i,j=1, \dots, d$ and  $k=0,\dots , T$ when the inverse matrix $F_{k}^{-1}$ of $F_{k}$ exists.
   \par 
   In the following we will denote by $M_{k; i, j}$ the matrix $F_{k}$ without the $i$-th row and $j$-th column. Recall also from linear algebra that if the inverse of a symmetric matrix $F_{k}$ exists then $F_{k; j, i}^{-1} = \frac{ (-1)^{i+j} \det(M_{k; i, j}) }{ \det(F_{k}) }$ which we use in Lemma \ref{lem:the alpha, beta, gamma terms}. 
   \begin{lem}
     \label{lem:F-matrix}
     For all $d \in \mathbb{N}_{\geq 2}$:
     \begin{align}
       \det(M_{k; i, j})^{2}
       & \leq C F_{k; j, j}^{0} F_{k; i, i}^{0} \prod_{\substack{ l=1\\ l \neq i,j}}^{d} | F_{k; l, l} |^{2}
       \quad && \text{ for all } i, j = 1, \dots, d \text{ with } i \neq j
       \label{eq:lemma F-matrix ineq 1}
       \\ 
       | F_{k; j, j} |^{2} \det(M_{k; j, j})^{2}
       & \leq \tilde C \det(F_{k}^{A})^{2}
       \quad && \text{ for all } j = 1, \dots, d
       \label{eq:lemma F-matrix ineq 2}
       \\ 
       F_{k; j, j} F_{k; i, i} \det(M_{k; i, j})^{2}
       & \leq \bar C \det(F_{k}^{A})^{2}
       \quad && \text{ for all } i, j = 1, \dots, d
       \label{eq:lemma F-matrix ineq 3}
     \end{align}
     for some positive constants $C, \tilde C$ and $\bar C$ where $F_{k}^{A} := \diag(A_{k; 1}, \dots, A_{k; d})$.
   \end{lem}
   \begin{proof}
     First note that the last inequality (\ref{eq:lemma F-matrix ineq 3}) follows from the first two. Indeed for the case $i \neq j$ and since $F_{k; j, j}^{0} \leq F_{k; j, j}$ (since $\varepsilon_{k+1}^{j}$ and $S_{k+1}^{j}$ are non-negative) for all $j$, then from inequality (\ref{eq:lemma F-matrix ineq 1}) we have 
     \begin{equation}
       \det(M_{k; i, j})^{2}
       \leq C F_{k; j, j} F_{k; i, i} \prod_{\substack{ l=1\\ l \neq i,j}}^{d} | F_{k; l, l} |^{2}   \,.
     \end{equation}
   Since the matrix $F_{k}^{A}$ is a diagonal matrix then it is clear that now inequality (\ref{eq:lemma F-matrix ineq 3}) follows for $i \neq j$. The case $i=j$ follows directly from inequality (\ref{eq:lemma F-matrix ineq 2}). 
   \par
   For showing the inequalities (\ref{eq:lemma F-matrix ineq 1}) and (\ref{eq:lemma F-matrix ineq 2}) for $d=2$ is trivial. We will show for the case $d=3$ the inequality (\ref{eq:lemma F-matrix ineq 1}). Inequality (\ref{eq:lemma F-matrix ineq 2}) follows then analogously. Let w.l.o.g. $i=1$. For $j=2$ we have 
   \begin{align}
     \det(M_{k; 1, 2})^{2}
     = (D_{k; 1, 2} A_{k; 3} - D_{k; 2, 3} D_{k; 1, 3})^{2} 
     \leq 2 | D_{k; 1, 2} |^{2} | A_{k; 3} |^{2} + 2 | D_{k; 2, 3} |^{2} | D_{k; 1, 3} |^{2}
   \end{align}
   where we have used the inequality $(a + b)^{2} \leq 2 a^{2} + 2 b^{2}$. Now, applying the conditional Cauchy-Schwarz inequality we get, 
   \begin{align}
     \det(M_{k; 1, 2})^{2}
     \leq 2 A_{k; 1}^{0} A_{k; 2}^{0} | A_{k; 3} |^{2} + 2 A_{k; 2}^{0} A_{k; 3}^{0} A_{k; 1}^{0} A_{k; 3}^{0} 
     \leq 4 A_{k; 1}^{0} A_{k; 2}^{0} | A_{k; 3} |^{2}   \,.
   \end{align}
   The case $j=3$ follows analogously and so inequality (\ref{eq:lemma F-matrix ineq 1}) holds.
   \par
   A generalization of the proof for an arbitrary $d$ can be done using the Laplace's formula and the symmetry of the matrices $F_{k}$ and $F_{k}^{0}$.
   \end{proof}
   The next Definition of the \textit{$F$-property} is crucial, not only for extending the LRM-criterion of \cite{schweizer:1988} to the illiquid case (i.e. $\varepsilon \neq 0$) but also (especially) for the extension to the multidimensional case. In the $1$-dimensional case the $F$-property translates to $\condVar{ \Delta S_{k+1} } + \condE{ \varepsilon_{k+1}S_{k+1} } \geq 0$ for a one-dimensional price process $S$ which is always fulfilled.\footnote{Recall the assumption that the price process $S$ and the process $\varepsilon$ are both non-negative.} Also if we are dealing with independent components, i.e., $S^{i}$ and $S^{j}$ are independent for $i \neq j$, then it reduces to $\det(F_{k}^{A}) \geq 0$ which also always holds since the matrix $F_{k}^{A}$ is positive semi-definite. So the next property is essentially linked to the covariance matrix of the multidimensional price process $S$. We will see later on in Section \ref{sec:the reduction of the F-property to the covariance matrix} that this property can be reduced to a property on the covariance matrix of $S$. In what follows, $C$ denotes a generic positive constant that might change from line to line.
   \begin{defi}
     \label{defi:the F-property}
     We say that the process $S$ has the $F$-property if there exists some $\delta \in (0, 1)$ such that
     \begin{equation}
       \det(F_{k}) - (1 - \delta) \det(F_{k}^{A}) 
       \geq 0
     \end{equation}
     for all $k = 0, 1, \dots, T$ where $F_{k}^{A} := \diag(A_{k; 1}, \dots, A_{k; d})$.
   \end{defi}
    \begin{lem}
     \label{lem:alpha beta terms boundedness} 
     Assume that $S$ has the $F$-property and satisfies the $F$-diagonal condition. Then the terms $\alpha_{k;i,j}$, $\beta_{k;i,j}$, $\alpha_{k;i,j}^{\varepsilon}$ and $\beta_{k;i,j}^{\varepsilon}$ are uniformly bounded in $k$ and $\omega$ for all $i,j = 1, \dots, d$. 
   \end{lem}
   \begin{proof}
      For the first term $\alpha_{k;i,j}$ we have
      \begin{align}
        \alpha_{k;i,j}
        = F_{k;  j,  j}^{0}  F_{k; i, i}^{0} \frac{ \det(M_{k; i, j})^{2} }{ \det(F_{k})^{2} }
        \leq C \frac{ \det(F_{k}^{A})^{2} }{ \det(F_{k})^{2} }
        \leq C \frac{ 1 }{ (1 - \delta)^{2} }
      \end{align}
      by using first the inequality (\ref{eq:lemma F-matrix ineq 3}) from Lemma \ref{lem:F-matrix} and then the $F$-property. For the second term $\beta_{k;i,j}$ we can estimate for the case $i = j$
      \begin{align}
        \beta_{k;i,j}
        = F_{k; i, i}^{0} \frac{ \det(M_{k; i, i})^{2} }{ \det(F_{k})^{2} }
        \leq C \frac{ F_{k; i, i}^{0} }{ | F_{k; i, i} |^{2} }  \frac{ \det(F_{k}^{A})^{2} }{ \det(F_{k})^{2} }
        \leq C \frac{ 1 }{ (1 - \delta)^{2} }
      \end{align}
      using inequality (\ref{eq:lemma F-matrix ineq 2}) from Lemma \ref{lem:F-matrix} and then the $F$-property and inequality (\ref{eq:diagonal condition F ineq 1}). For the case $i \neq j$ and using inequality (\ref{eq:lemma F-matrix ineq 1}) from Lemma \ref{lem:F-matrix}  we have
      \begin{align}
         \det(F_{k})^{2} \beta_{k;i,j}
        = F_{k; i, i}^{0} \det(M_{k; i, j})^{2} 
        \leq C F_{k; j, j}^{0} | F_{k; i, i}^{0} |^{2} \prod_{\substack{ l=1\\ l \neq i,j}}^{d} | F_{k; l, l} |^{2}
        \leq C \frac{ F_{k; j, j}^{0} }{ | F_{k; j, j} |^{2} } \det(F_{k}^{A})^{2}
      \end{align} 
      and from the $F$-property and inequality (\ref{eq:diagonal condition F ineq 1}), $\beta_{k;i,j}$ is uniformly bounded. Furthermore and by the same arguments as for the term $\beta_{k;i,j}$ we have for the case $i=j$
      \begin{align}
        \alpha_{k;i,j}^{\varepsilon}
        = | F_{k; i, i}^{\varepsilon} |^{2} F_{k; i, i}^{0} \frac{ \det(M_{k; i, i})^{2} }{ \det(F_{k})^{2} }
        \leq C | F_{k; i, i}^{\varepsilon} |^{2} \frac{ F_{k; i, i}^{0} }{ | F_{k; i, i} |^{2} }  \frac{ \det(F_{k}^{A})^{2} }{ \det(F_{k})^{2} }
        \leq C \frac{ 1 }{ (1 - \delta)^{2} }
      \end{align}
      using the $F$-property and inequality (\ref{eq:diagonal condition F ineq 2}). For $i \neq j$ we can estimate
      \begin{align}
        \det(F_{k})^{2} \alpha_{k;i,j}^{\varepsilon}
        = | F_{k; i, i}^{\varepsilon} |^{2} F_{k; j, j}^{0} \det(M_{k; i, j})^{2} 
        &\leq C  | F_{k; i, i}^{\varepsilon} |^{2} F_{k; i, i}^{0} | F_{k; j, j}^{0} |^{2} \prod_{\substack{ l=1\\ l \neq i,j}}^{d} | F_{k; l, l} |^{2}
        \nonumber  \\ 
        &\leq C | F_{k; i, i}^{\varepsilon} |^{2} \frac{ F_{k; i, i}^{0} }{ | F_{k; i, i} |^{2} } \det(F_{k}^{A})^{2}
      \end{align}
      and from the $F$-property and inequality (\ref{eq:diagonal condition F ineq 2}), $\alpha_{k;i,j}^{\varepsilon}$ is also uniformly bounded. For the last term $\beta_{k;i,j}^{\varepsilon}$ we have for $i = j$  
      \begin{align}
        \beta_{k;i,j}^{\varepsilon}
        = | F_{k; i, i}^{\varepsilon} |^{2} \frac{ \det(M_{k; i, i})^{2} }{ \det(F_{k})^{2} }
        \leq C \frac{ | F_{k; i, i}^{\varepsilon} |^{2} }{ | F_{k; i, i} |^{2} }  \frac{ \det(F_{k}^{A})^{2} }{ \det(F_{k})^{2} }
        \leq C \frac{ 1 }{ (1 - \delta)^{2} }
      \end{align}
      by the $F$-property. Moreover for $i \neq j$
      \begin{align}
        \det(F_{k})^{2} \beta_{k;i,j}^{\varepsilon}
        = | F_{k; i, i}^{\varepsilon} |^{2} \det(M_{k; i, j})^{2} 
        &\leq C  | F_{k; i, i}^{\varepsilon} |^{2} F_{k; i, i}^{0} F_{k; j, j}^{0}  \prod_{\substack{ l=1\\ l \neq i,j}}^{d} | F_{k; l, l} |^{2}
        \nonumber   \\ 
        &= C | F_{k; i, i}^{\varepsilon} |^{2} \frac{ F_{k; i, i}^{0} }{ | F_{k; i, i} |^{2} } \frac{ F_{k; j, j}^{0} }{ | F_{k; j, j} |^{2} } \det(F_{k}^{A})^{2}
      \end{align}
      where from the $F$-property and the $F$-diagonal condition the last term $\beta_{k;i,j}^{\varepsilon}$ is uniformly bounded. We also made use of the fact that the process $\varepsilon$ is deterministic and that we have a finite number of hedging times.
   \end{proof}
   \begin{lem}
     \label{lem:the alpha, beta, gamma terms}
     Assume that $F_{k}^{-1}$ exists for $k \in \{ 0, 1, \dots, T \}$ and $S$ has bounded mean-variance tradeoff. Let $(X,Y)$ be any trading strategy. Then there exists some constant $C > 0$ such that
     \begin{eqnarray}
       && \E[ ( (F_{k}^{-1} b_{k})_{j} \Delta S_{k+1}^{j} )^{2} ] \label{strat:bound:1} \nonumber\\
       &\leq &C \E[  \condVar{ V_{k+1} } \sum_{i=1}^{d} \alpha_{k;i,j} 
        + \sum_{i=1}^{d}( c(\varepsilon_{k+1}) \alpha_{k;i,j} + \alpha_{k;i,j}^{\varepsilon}) \condE{ | X_{k+2}^{i} |^{2} }  ]
       \\
      && \E[ ( (F_{k}^{-1} b_{k})_{j} )^{2} ] \nonumber \\
       &\leq & C \E[  \condVar{ V_{k+1} } \sum_{i=1}^{d} \beta_{k;i,j}
        + \sum_{i=1}^{d}( c(\varepsilon_{k+1}) \beta_{k;i,j} + \beta_{k;i,j}^{\varepsilon}) \condE{ | X_{k+2}^{i} |^{2} }  ] \label{strat:bound:2}
     \end{eqnarray}
     for all $j = 1, \dots, d$ where $(F_{k}^{-1} b_{k})_{j}$ is the $j$-th component of the vector $(F_{k}^{-1} b_{k})$. The term $c(\varepsilon_{k+1})$ denotes a positive constant depending on the process $\varepsilon$ at time $k+1$ such that for $\varepsilon_{k+1} \to 0$, $c(\varepsilon_{k+1})$ converges to zero.
   \end{lem}
   \begin{proof}
   First note that from the definition of the variance and using bounded mean-variance tradeoff, it follows directly that 
   \begin{equation}
   \label{eq:1 lem:the alpha, beta, gamma terms}
       \condE{ | \Delta S_{k+1}^{j} |^{2} }
       = \condVar{ \Delta S_{k+1}^{j} } + ( \condE{ \Delta S_{k+1}^{j} } )^{2}
       \leq C A_{k; j}^{0}     \,.
     \end{equation}
  Furthermore, denoting $F = F_{k}$ and $b = b_{k}$ we have from the tower property and using inequality (\ref{eq:1 lem:the alpha, beta, gamma terms})
   \begin{align}
     \E[ ( (F^{-1} b)_{j} \Delta S_{k+1}^{j} )^{2} ]
     &= \E[ ( (F^{-1} (b^{0} + b^{\varepsilon}))_{j} )^{2} \condE{ | \Delta S_{k+1}^{j} |^{2} } ]
     \nonumber  \\
     &\leq 2C \E[ \sum_{i=1}^{d} | F_{j, i}^{-1} |^{2} (| b_{i}^{0} |^{2} + | b_{i}^{\varepsilon} |^{2}) F_{j, j}^{0}  ]
   \end{align}
   Moreover, using the conditional Cauchy-Schwarz-Inequality for the term $b_{i}^{0}$ and the conditional inequality $(\E[ XY | \mathcal{G} ])^{2} \leq \E [X^{2} | \mathcal{G} ] \E [Y^{2} | \mathcal{G} ]$ on the term $b_{i}^{\varepsilon}$ together with the definition of the variance yields
   \begin{align}
     &\E[ ( (F^{-1} b)_{j} \Delta S_{k+1}^{j} )^{2} ]
     \nonumber  \\
     &\leq C \E[ \sum_{i=1}^{d} | F_{j, i}^{-1} |^{2} (  \condVar{ V_{k+1} } F_{i, i}^{0} +  \condE{ | \varepsilon_{k+1}^{i} S_{k+1}^{i} |^{2} } \condE{ | X_{k+2}^{i} |^{2} } ) F_{j, j}^{0}  ]
     \nonumber  \\
     &=  C \E[ \sum_{i=1}^{d} | F_{j, i}^{-1} |^{2} (  \condVar{ V_{k+1} } F_{i, i}^{0} +  | \varepsilon_{k+1}^{i} |^{2} F_{i, i}^{0} \condE{ | X_{k+2}^{i} |^{2} } + | F_{i, i}^{\varepsilon} |^{2} \condE{ | X_{k+2}^{i} |^{2} } ) F_{j, j}^{0}  ]    \,.
   \end{align}
   The other inequality follows analogously. 
   \end{proof}
   \begin{remark}
For the Existence of a LRM-strategy under illiquidity we will use Lemma \ref{lem:alpha beta terms boundedness} together with Lemma \ref{lem:the alpha, beta, gamma terms}. For the optimal strategy $\hat X$ (under the LRM-criterion under illiquidity) we will need to show that $\hat X_{k+1}^{j} \Delta S_{k+1}^{j} \in \mathbb{L}_{T}^{2, 1}$ and $\hat X_{k+1}^{j} \in \mathbb{L}_{T}^{2, 1}$. The first integrability property shows that the strategy $\hat X$ belongs to $ \hat\Theta_{d}(S)$, the space of all $\mathbb{R}^{d}$-valued predictable strategies $X = (X_{k})_{k=1,2, \dots, T+1}$ so that $ X_{k}^{*} \Delta S_{k} \in \mathbb{L}_{T}^{2, 1}$ for $k=1, 2, \dots, T$. The second one is needed to show the first one. Nevertheless, both integrability properties are needed in order to show that the liquidity costs of the optimal strategy are integrable. 
\par 
In the infinite liquidity case, that is $\varepsilon = 0$, since the terms $c(\varepsilon_{k+1}) \alpha_{k;i,j}$ and $\alpha_{k;i,j}^{\varepsilon}$ vanish,  we do not need the second inequality of Lemma \ref{lem:the alpha, beta, gamma terms}. This implies that in the multidimensional case without liquidity costs, one needs to show only that $\hat X \in \hat \Theta_{d}(S)$ by using bounded mean-variance tradeoff and the $F$-property.
\par
Also, in the $1$-dimensional case ($d=1$) we have 
\begin{equation}
\alpha_{k;1,1} = \frac{ | A_{k; 1}^{0} |^{2} }{ | A_{k; 1} |^{2} } 
\quad , \quad
\beta_{k;1,1} = \frac{A_{k;1}^{0}}{| A_{k;1} |^{2}} 
\quad , \quad
\alpha_{k;1,1}^{\varepsilon} = A_{k; 1}^{\varepsilon}  \frac{ A_{k; 1}^{0} }{ | A_{k; 1} |^{2} } 
\quad , \quad
\beta_{k;1,1}^{\varepsilon} =  \frac{ | A_{k; 1}^{\varepsilon} |^{2} }{ | A_{k; 1} |^{2} } 
\end{equation}
where the terms $\alpha_{k;1,1}$, $\beta_{k;1,1}^{\varepsilon}$ are bounded by $1$ and the terms $\beta_{k;1,1}$, $\alpha_{k;1,1}^{\varepsilon}$ are uniformly bounded by the $F$-diagonal property. Moreover for $\varepsilon = 0$ one would only need to show the first inequality of Lemma \ref{lem:the alpha, beta, gamma terms} which reduces to 
\begin{equation}
  \E[ ( (F_{k}^{-1} b_{k})_{1} \Delta S_{k+1}^{1} )^{2} ]
       \leq C \E[ | V_{k+1} |^{2} ]
\end{equation}
as in the classical $1$-dimensional case in \cite{schweizer:1988}. Recall that in this case only the assumption of bounded mean-variance tradeoff is essential. 
   \end{remark}
  We continue with the main Theorem where we show the existence of a local risk-minimizing strategy under illiquidity and under some mild conditions on the marginal price process $S$.

   %The Main Theorem
  \subsection{Existence and recursive construction of an optimal strategy}
  \label{sec:the main theorem}

  Using the assumptions imposed in the previous Section \ref{sec:assumptions} we are able to prove the existence of a local risk-minimizing strategy under illiquidity and additionally to give an explicit representation by means of a backward induction argument.
  \begin{theo}[\textbf{Existence result}]
    \label{theo:existence}
    Assume that $S$ has the $F$-property, bounded mean-variance tradeoff and satisfies the $F$-diagonal condition. Let further the covariance matrix $F_{k}^{0}$ be positive definite at all times $k = 0, 1, \dots, T-1$. Then for any contingent claim $H = \bar X_{T+1}^{*} S_{T} + \bar Y_{T} \in \mathbb{L}_{T}^{2, 1}$ with $\bar X_{T+1}^{*} S_{T} \in \mathbb{L}_{T}^{2, 1}$ and $\bar X_{T+1} \in \mathbb{L}_{T}^{2, d}$, there exists a local risk-minimizing strategy $\hat \varphi = ( \hat X, \hat Y )$ under illiquidity with $\hat X_{T+1} = \bar X_{T+1}$ and $\hat Y_{T} = \bar Y_{T}$. Furthermore, the strategy has the representation
    \begin{align}
      \hat X_{k+1} 
      &= F_{k}^{-1} b_{k}
      \quad \Pas \text{ for } k = 0, \dots, T-1
      \label{eq:LRM-strategy X}
      \\
      \hat Y_{k}
      &= \condE{ \hat W_{k} } - \hat X_{k+1}^{*} S_{k}
      \quad \Pas \text{ for } k = 0, 1, \dots, T-1
      \label{eq:LRM-strategy Y}
    \end{align}
    where $\hat W_{k} = H - \sum_{m=k+1}^{T} \hat X_{m}^{*} \Delta S_{m}$.
  \end{theo}
  \begin{proof}
  The proof is a backward induction argument on $k = 0, 1, \dots, T$. First set $\hat X_{T+1} = \bar X_{T+1}$ and $\hat Y_{T} = \bar Y_{T}$. So, fix some $k \in \{ 0, 1, \dots, T-2 \}$ and assume that at times $l = k, \dots, T-2$
  \begin{enumerate}[label=(\roman*)]
    \item \label{item:theorem proof step 1} $ \hat X_{l+2}^{j} \Delta S_{l+2}^{j} \in \classltwoone \quad \text{and} \quad \hat X_{l+2}^{j} \in \classltwoone $
    \item \label{item:theorem proof step 2} $| \hat X_{l+2}^{j} |^{2} S_{l+1}^{j} \in \classloneone$
    \item \label{item:theorem proof step 3} $ \hat X_{l+2}^{*} S_{l+1} + \hat Y_{l+1}  \in \classltwoone
               \quad , \quad \hat Y_{l+1} \in \mathcal{F}_{l+1} $
  \end{enumerate}
  for all $j = 1, \dots, d$ holds. At time $k$ we want to minimize the expression (\ref{eq:minimizing expression}) over all $X_{k+1}^{'}$ and show that the following properties are fulfilled for all $j = 1, \dots, d$:
   \begin{enumerate}[label=(\roman*)]
    \item \label{item:theorem proof step 1} $ X_{k+1}^{', j} \Delta S_{k+1}^{j} \in \classltwoone \quad \text{and} \quad X_{k+1}^{', j} \in \classltwoone $
    \item \label{item:theorem proof step 2} $ | X_{k+1}^{', j} |^{2} S_{k}^{j} \in \classloneone $
    \item \label{item:theorem proof step 3} $ (X_{k+1}^{'})^{*} S_{k} + Y_{k}^{'} \in \classltwoone
               \quad , \quad Y_{k}^{'} \in \mathcal{F}_{k} $
  \end{enumerate}
  Properties \ref{item:theorem proof step 1} - \ref{item:theorem proof step 3} will then ensure that $(\hat X, \hat Y) \in \Theta_d (S)$. First we define the function $f_{k}$ as in equation (\ref{eq:minimizing function}) and note that all the terms in $f_{k}$ are integrable by induction hypothesis. Since $F_{k}$ is positive definite then there exists a unique solution to the minimization problem and an $\mathcal{F}_{k}$-measurable minimizer $\hat X_{k+1}$ can be constructed, which equals $F_{k}^{-1} b_{k}$. Furthermore define $\hat Y_{k}$ as in equation (\ref{eq:LRM-strategy Y}). Then it is clear that $\hat Y_{k}$ is $\mathcal{F}_{k}$-measurable. The fact that $\hat X_{k+1}^{*} S_{k} + \hat Y_{k} = \condE{ \hat W_{k} } \in \classltwoone$ follows from $H \in \classltwoone$, the induction hypothesis $\sum_{m=k+2}^{T} \hat X _{m}^{*} \Delta S_{m} \in \classltwoone$ and $\hat X _{k+1}^{*} \Delta S_{k+1} \in \classltwoone$, which we will show below. 
   \par
   Now let us show first that $ \hat X_{k+1}^{j} \Delta S_{k+1}^{j} \in \classltwoone$. By inequality (\ref{strat:bound:1}) of Lemma \ref{lem:the alpha, beta, gamma terms} we know that for a constant $C>0$,
  \begin{eqnarray}
    & &   \E[ ( \hat X_{k+1}^{j} \Delta S_{k+1}^{j} )^{2} ]         \nonumber       \\
       &\leq & C \E[  \condVar{ \hat X_{k+2}^{*} S_{k+1} + \hat Y_{k+1} } \sum_{i=1}^{d} \alpha_{k;i,j} 
        + \sum_{i=1}^{d}( c(\varepsilon_{k+1}) \alpha_{k;i,j} + \alpha_{k;i,j}^{\varepsilon}) \condE{ | X_{k+2}^{i} |^{2} }  ]
   \end{eqnarray}
 holds. Since by the induction hypothesis $\hat X_{k+2}^{*} S_{k+1} + \hat Y_{k+1}$ and $\hat X_{k+2}^{i}$ both in $\classltwoone$ for all $i=1, \dots, d$, then it remains to show that the terms $\alpha_{k;i,j}$, $\alpha_{k;i,j}^{\varepsilon}$ are uniformly bounded in $k$ and $\omega$. This follows from Lemma \ref{lem:alpha beta terms boundedness}. Similarly one can show that $ \hat X_{k+1}^{j} \in \classltwoone$ using inequality (\ref{strat:bound:2}) of Lemma \ref{lem:the alpha, beta, gamma terms}.
  \par
  Next we show that the liquidity costs $\condE{ \sum_{j=1}^{d} \varepsilon_{k+1}^{j} S_{k+1}^{j} | \hat X_{k+2}^{j} - \hat X_{k+1}^{j} |^{2} } $ are integrable. From the minimization problem of expression (\ref{eq:minimizing expression}) and since $\hat X_{k+1}$ is a minimizer, we know that (w.l.o.g. $\alpha = 1$):
  \begin{align}
    & \condVar{ \hat X_{k+2}^{*} S_{k+1} + \hat Y_{k+1} - (\hat X_{k+1})^{*} \Delta S_{k+1} }
       + \condE{ \sum_{j=1}^{d} \varepsilon_{k+1}^{j} S_{k+1}^{j} | \hat X_{k+2}^{j} - \hat X_{k+1}^{j} |^{2} }  
     \nonumber     \\
     & \leq \condVar{ \hat X_{k+2}^{*} S_{k+1} + \hat Y_{k+1}}
                + \condE{ \sum_{j=1}^{d} \varepsilon_{k+1}^{j} S_{k+1}^{j} | \hat X_{k+2}^{j} |^{2} }        
  \end{align}
  holds, where the right hand side corresponds to choosing $X_{k+1}=0$.
  Taking expectation on both sides and since by definition the conditional variance is non-negative, we get
  \begin{align}
    \E[ \sum_{j=1}^{d} \varepsilon_{k+1}^{j} S_{k+1}^{j} | \hat X_{k+2}^{j} - \hat X_{k+1}^{j} |^{2} ]
    \leq \E[ | \hat X_{k+2}^{*} S_{k+1} + \hat Y_{k+1} |^{2} ]
          + \E[ \sum_{j=1}^{d} \varepsilon_{k+1}^{j} S_{k+1}^{j} | \hat X_{k+2}^{j} |^{2} ]  
  \end{align}
  where we have used the fact that $\Var(X) \leq \E|X|^{2}$. Now, since by the inductive hypothesis,  $\hat X_{k+2}^{*} S_{k+1} + \hat Y_{k+1} \in \classltwoone$ and $S_{k+1}^{j} | \hat X_{k+2}^{j} |^{2} \in \classloneone$ for all $j=1, \dots, d$ then it is clear that the liquidity cost $\sum_{j=1}^{d} \varepsilon_{k+1}^{j} S_{k+1}^{j} | \hat X_{k+2}^{j} - \hat X_{k+1}^{j} |^{2}$ is in $\classloneone$. In particular $ \varepsilon_{k+1}^{j} S_{k+1}^{j} | \hat X_{k+2}^{j} - \hat X_{k+1}^{j} |^{2} \in \classloneone$ for all $j=1, \dots, d$. This holds from the fact that the deterministic process $\varepsilon$ and the marginal price process $S$ are both non-negative by assumption.
  \par 
    In order to complete the proof, it remains to show that $| \hat X_{k+1}^{j} |^{2} S_{k}^{j} \in \classloneone$. This is needed in order to complete the induction argument and be able to show that the liquidity costs in the next step are again integrable. So, from the equality 
    \begin{equation}
      | \hat X_{k+1}^{j} |^{2} S_{k}^{j}
      = - | \hat X_{k+1}^{j} |^{2} \Delta S_{k+1}^{j} + | \hat X_{k+1}^{j} |^{2} S_{k+1}^{j}
    \end{equation}
    we need to show that $| \hat X_{k+1}^{j} |^{2} \Delta S_{k+1}^{j}$ and $| \hat X_{k+1}^{j} |^{2} S_{k+1}^{j}$ are both in $\classloneone$.
    Since, as already shown, the liquidity costs are integrable for all $j=1, \dots, d$ and since by induction hypothesis $ | \hat X_{k+2}^{j} |^{2} S_{k+1}^{j} \in \classloneone$ then the inequality 
    \begin{align}
      0  
      \leq | \hat X_{k+1}^{j} |^{2} S_{k+1}^{j}
      \leq 2 | \hat X_{k+2}^{j} - \hat X_{k+1}^{j} |^{2} S_{k+1}^{j} + 2 | \hat X_{k+2}^{j} |^{2} S_{k+1}^{j}
    \end{align}
 follows. Since $\varepsilon_{k}^j >0$ this implies that $| \hat X_{k+1}^{j} |^{2} S_{k+1}^{j}$ is integrable for all $j=1, \dots, d$. %The case $\varepsilon_{k}^j =0$ can be treated seperately. 
 The term $| \hat X_{k+1}^{j} |^{2} \Delta S_{k+1}^{j}$ is also integrable by the fact that $\hat X_{k+1}^{j} \Delta S_{k+1}^{j}$ and $\hat X_{k+1}^{j}$ are both in $\classltwoone$. Indeed we have
    \begin{align}
       \E[ | \hat X_{k+1}^{j} |^{2} \Delta S_{k+1}^{j} ]  
       &\leq \E[ | \hat X_{k+1}^{j} |^{2} \mathbf{1}_{ \{ | \Delta S_{k+1}^{j} | \leq 1 \} } ]
       + \E[ | \hat X_{k+1}^{j} \Delta S_{k+1}^{j} |^{2} \mathbf{1}_{ \{ | \Delta S_{k+1}^{j} | \geq 1 \} } ]
        \nonumber    \\
       & \leq \E[ | \hat X_{k+1}^{j} |^{2} ]
       + \E[ | \hat X_{k+1}^{j} \Delta S_{k+1}^{j} |^{2} ]
     \end{align}
  and this proves and completes the induction step at time $k$.
  \par
  The base case at time $k=T$ where $\hat X_{T+1}^{*} S_{T} + \hat Y_{T} = H$ is clear by the same arguments and by the assumptions on $H$ and $\bar X_{T+1}$, $\bar Y_{T}$. Indeed, since $\hat X_{T+1}^{*} S_{T} + \hat Y_{T}$ and $\hat X_{T+1}$ are both square integrable, then from Lemma \ref{lem:the alpha, beta, gamma terms} and Lemma \ref{lem:alpha beta terms boundedness} it follows that $\hat X_{T}^{j} \Delta S_{T}^{j} \in \classltwoone$ and $\hat X_{T}^{j} \in \classltwoone$ for all $j$. Moreover, note that with the assumptions $\hat X_{T+1}^{j} S_{T}^{j} \in \classltwoone$, $\hat X_{T+1}^{j} \in \classltwoone$ one can show that $| \hat X_{T+1}^{j} |^{2} S_{T}^{j} \in \classloneone$. By the same arguments as above, this will imply the integrability of the liquidity costs. The fact that $| \hat X_{T}^{j} |^{2} S_{T-1}^{j} \in \classloneone$ can be shown by using exactly the same arguments as in the proof for the inductive step.
  \par
  Finally, by defining 
  \begin{equation}
    \hat Y_{T-1}
      = \condE{ H - \hat X_{T}^{*} \Delta S_{T} }- \hat X_{T}^{*} S_{T-1}
  \end{equation}
  then it is clear that $\hat Y_{T-1}$ is $\mathcal{F}_{T-1}$-measurable and $\hat X_{T}^{*} S_{T-1} + \hat Y_{T-1} = \condE{ H - \hat X_{T}^{*} \Delta S_{T} } $ belongs to $\classltwoone$.
  \par
  The martingale property of $C(\hat \varphi)$ follows from the construction of $\hat Y$ since at each time $k$ we have
  \begin{equation}
    \condE{ C_{T}(\hat \varphi) - C_{k}(\hat \varphi) }
    = 0
  \end{equation} 
  and so by Proposition \ref{prop:minimizing variance}, since both properties are satisfied, then the trading strategy $\hat \varphi = ( \hat X, \hat Y )$ is local risk-minimizing under illiquidity and the proof is complete.
  \end{proof}
  \begin{remark}
  \label{rem:1-dim strategy and the book values}
    In the $1$-dimensional case, the LRM-strategy $\hat \varphi = ( \hat X, \hat Y )$ under illiquidity has the representation
    \begin{align}
      \hat X_{k+1} 
      &= \frac{ \condCov{ V_{k+1}(\hat \varphi), \Delta S_{k+1} }
                   + \condE{ \varepsilon_{k+1} S_{k+1} \hat X_{k+2} } }
                 { \condVar{ \Delta S_{k+1} } + \condE{ \varepsilon_{k+1} S_{k+1} } } 
       \\
       V_{k}(\hat \varphi)
       &= \E \left [ H - \sum_{m=k+1}^{T} \hat X_{m} \Delta S_{m} \middle|\ \mathcal{F}_k \right]
    \end{align}
    For $\varepsilon_{k+1}$ tending to zero we get the classical local risk minimization strategy without accounting for illiquidity. Let us denote this by $\bar \varphi = ( \bar X, \bar Y )$. Also, one can easily note that in the case where $S$ is a martingale, then 
    $V_{k}(\hat \varphi) = \condE{ H } = V_{k}(\bar \varphi)$. That means the two book values are equal. 
    %\panos{(Interpretation=?)} This can be interpreted by the fact that with milder strategies we can get the same book values since by the LRM-criterion we adjusting the strategy at each time without being self-financing but only with the (classical) cost process being a martingale.
    \par
    One can easily check that when $\varepsilon_{k+1}$ goes to infinity, i.e. infinite liquidity costs, then 
    \begin{equation}
      \hat X_{k+1}
      \rightarrow 
      \E \left[ \frac{ S_{k+1} \cdots S_{T} \hat X_{T+1} }{ \mathbb{E}[ S_{k+1} | \mathcal{F}_{k} ] \cdots \mathbb{E}[ S_{T} | \mathcal{F}_{T-1} ] } \middle|\ \mathcal{F}_k \right].
    \end{equation}
    Consider \textit{cash settlement}, i.e. $\hat X_{T+1} = 0$ and $\hat Y_{T} = H$, where the value of the option has to be paid out in cash as it is usually the market standard. Then we clearly have $\hat X_{k+1} \rightarrow 0$ for all $k = 0, 1, \dots, T$ when $\varepsilon_{k+1} \rightarrow \infty$. From a financial point of view this makes sense since for the investor the best choice is to invest nothing to avoid infinite liquidity cost. A similar observation can be made in the $d$-dimensional case.  
  \end{remark}
%  \begin{remark}
%    Note that the assumption on the covariance matrix $F_{k}^{0}$ that it is positive definite is important since that gives us that for all $j=1, \dots, d$ we have $\condVar{ \Delta S_{k+1}^{j} } > 0$ at all times $k=0, 1, \dots, T$. If for example in the $1$-dimensional case we assume that $F_{k}$ is positive definite and it happens to have $\condVar{ \Delta S_{k+1}^{j} } = 0$ then we can still apply the existence theorem\footnote{That also means $\varepsilon_{k+1} > 0$.} and we get $\hat X_{k+1} = \condE{ \hat X_{k+2} }$ which does not depend on the liquidity parameter $\varepsilon$. Furthermore, having $S_{k+1}, \dots, S_{T}$ equal to a constant then we have the expression
%    \begin{equation*}
%      \hat X_{k+1} = \mathbb{E}[ \hat X_{T+1} | \mathcal{F}_{k} ]
%    \end{equation*}
%    which make sense only for a cash settlement option i.e.,  $\hat X_{T+1} = 0$.
%  \end{remark}

  %The reduction of the F-property to the covariance matrix
  \subsection{A sufficient condition for the $F$-property in terms of the covariance matrix $F^{0}$}
  \label{sec:the reduction of the F-property to the covariance matrix}

  Recall that the $F$-property from Definition \ref{defi:the F-property} was used in order to show the integrability properties of Proposition~\ref{prop:minimizing variance} for the local risk-minimizing strategy under illiquidity calculated backwards in time in the proof of Theorem~\ref{theo:existence}. In this section we show how this condition is related to the covariance matrix $F^{0}$. Before we continue let us recall the definition of a principal submatrix \cite[see][]{horn.johnson:2012}. 
  \par
  \textbf{Definition of a principal submatrix: } In general let $P \in \mathbb{R}^{m, n}$ be a real matrix with $m$ rows and $n$ columns, and let $\alpha \subset \{1, \dots, m  \}$, $\beta \subset \{1, \dots, n  \}$ be index sets. Denote by $P[ \alpha, \beta ]$ the (sub)matrix of entries that lie in the rows of $P$ indexed by $\alpha$ and the columns indexed by $\beta$.
 % $P$-entries with rows in $\alpha$ and with columns in $\beta$. 
  For $\alpha = \beta$ denote by $P[ \alpha ] = P[ \alpha, \alpha ]$ the (sub)matrix of entries that lie in the rows and columns of $P$ indexed by $\alpha$. 
  %$P$-entries with rows and columns in $\alpha$. 
  Then $P[ \alpha ]$ is called a \textit{principal submatrix} of $P$.\footnote{ A matrix $P \in \mathbb{R}^{n, n}$ has $\binom{n}{l}$ distinct principal submatrices of size $l \times l$. }
  \par
 The following Lemma \ref{lem:the reduced $F$-property} yields a sufficient criterion in terms of the covariance matrix $F^{0}$.
  \begin{lem}
    \label{lem:the reduced $F$-property}
    $S$ has the $F$-property if there exists some $\delta \in ( 0, 1 )$ such that
    \begin{equation}
    \label{eq:0 lem:the reduced $F$-property}
      \det(P_{k}^{0}) - (1 - \delta) \det(P_{k}^{A^{0}}) 
      \geq 0
    \end{equation}
    for all principal submatrices $P_{k}^{0}$ of $F_{k}^{0}$ and principal submatrices $P_{k}^{A^{0}}$ of $F_{k}^{A^{0}}$ where $F_{k}^{A^{0}} := \diag(A_{k; 1}^{0}, \dots, A_{k; d}^{0})$ of size $l \times l $ where $l \in \{ 2, \dots, d \}$ and for all $k = 0, 1, \dots, T$.
  \end{lem}
  %NEW PROOF FOR THE d-DIM CASE
  \begin{proof}
  Let $d \in \mathbb{N}_{\geq 2}$, fix $k \in \{ 0, 1, \dots, T \}$ and omitting the time $k$ denote $F = F_{k}$.
  \par
  Furthermore we denote by $F^{A_{m}^{0}; A_{l}} := F^{A_{m}^{0}; A_{l}}(A_{m}^{0}, A_{m+1}^{0}, \dots, A_{l-1}^{0}, A_{l+1}, A_{l+2}, \dots,  A_{d})$ for $m, l \in \{ 1, \dots, d \}$, $m < l$,  the $(d-m) \times (d-m)$ symmetric matrix where for $i=j$, $j \in \{ 1, \dots, l - m \}$ we set $F_{i, j}^{A_{m}^{0}; A_{l}} = A_{m + j - 1}^{0}$ and for $j \in \{ l-m, \dots, d - m - 1 \}$ we set $F_{i, j}^{A_{m}^{0}; A_{l}} = A_{m+ j +1}$ for the diagonal elements of the matrix. Otherwise for $i \neq j$ we set $F_{i, j}^{A_{m}^{0}; A_{l}} = D_{m + i - 1, m + j - 1}$ for $i, j \in \{ 1, \dots, l - m \}$ and $F_{i, j}^{A_{m}^{0}; A_{l}} = D_{m + i + 1, m + j + 1}$ for $i, j \in \{ l-m, \dots, d - m - 1 \}$. For $m = l$ we set $F^{A_{l}^{0}; A_{l}} := F^{A_{l}^{0}; A_{l}}(A_{l+1}, A_{l+2}, \dots,  A_{d})$ which is equal to $F$ without the first $l$ rows and columns. Also note that for $l = d$ we have $F^{A_{m}^{0}; A_{d}} := F^{A_{m}^{0}; A_{d}}(A_{m}^{0}, A_{m+1}^{0}, \dots,  A_{d - 1
 }^{0})$ which is equal to $F^{0}$ without the first $m - 1$ rows and columns and without the last row and the last column. 
  \par
  Since $A_{j} = A_{j}^{0} + A_{j}^{\varepsilon}$ and using the fact that the matrices $F$ and $F^{0}$ are symmetric then one can calculate that
  \begin{align}
  \label{eq:1 lem:the reduced $F$-property}
    &\det(F) - (1 - \delta) \det(F^{A})           \nonumber
    \\
    &= \det(F^{0}) - (1 - \delta) \det(F^{A^{0}})         \nonumber
    \\
    &+ A_{1}^{\varepsilon} [ \det(F^{A_{1}^{0}; A_{1}}(A_{2}, A_{3}, A_{4}, \dots,  A_{d})) 
                          - (1 - \delta) \det(\diag(A_{2}, A_{3}, A_{4}, \dots,  A_{d})) ]          \nonumber
    \\
    &+ A_{2}^{\varepsilon} [ \det(F^{A_{1}^{0}; A_{2}}(A_{1}^{0}, A_{3}, A_{4}, \dots,  A_{d})) 
                          - (1 - \delta) \det(\diag(A_{1}^{0}, A_{3}, A_{4}, \dots,  A_{d})) ]        \nonumber
    \\
    &+ A_{3}^{\varepsilon} [ \det(F^{A_{1}^{0}; A_{3}}(A_{1}^{0}, A_{2}^{0}, A_{4}, \dots,  A_{d})) 
                          - (1 - \delta) \det(\diag(A_{1}^{0}, A_{2}^{0}, A_{4}, \dots,  A_{d})) ]         \nonumber
    \\
    &+ \dots +       \nonumber
    \\
    &+ A_{d}^{\varepsilon} [ \det(F^{A_{1}^{0}; A_{d}}(A_{1}^{0}, A_{2}^{0}, A_{3}^{0}, \dots,  A_{d-1}^{0})) 
                          - (1 - \delta) \det(\diag(A_{1}^{0}, A_{2}^{0}, A_{3}^{0}, \dots,  A_{d-1}^{0})) ] \,, 
  \end{align}
  where $F^{0}$ is the $\binom{d}{d} = 1$ principal submatrix $P^{0}[ \{ 1, 2, \dots, d \} ]$ of size $d \times d$ and  $F^{A_{1}^{0}; A_{d}}(A_{1}^{0}, A_{2}^{0}, A_{3}^{0},$ \dots$,  A_{d-1}^{0}) = P^{0}[ \{ 1, 2, \dots, d - 1 \} ]$ one of the $\binom{d}{d-1} = d$ principal submatrices of $F^{0}$ of size $(d-1) \times (d-1)$. The remaining $d-1$ principal submatrices of size $(d-1) \times (d-1)$ can be calculated recursively as in equation (\ref{eq:1 lem:the reduced $F$-property}) for the $d-1$ terms in the R.H.S of the equation. For example we have 
   \begin{align}
    A_{1}^{\varepsilon} [ & \det(F^{A_{1}^{0}; A_{1}}(A_{2}, A_{3}, A_{4}, \dots,  A_{d})) 
                          - (1 - \delta) \det(\diag(A_{2}, A_{3}, A_{4}, \dots,  A_{d})) ]         
    \nonumber \\
    = 
    %A1
        A_{1}^{\varepsilon} \Big\{ 
                          &A_{2}^{\varepsilon} [ \det(F^{A_{2}^{0}; A_{2}}(A_{3}, A_{4}, A_{5}, \dots,  A_{d})) 
                          - (1 - \delta) \det(\diag(A_{3}, A_{4}, A_{5}, \dots,  A_{d})) ] 
    \nonumber \\
    + &A_{3}^{\varepsilon} [ \det(F^{A_{2}^{0}; A_{3}}(A_{2}^{0}, A_{4}, A_{5}, \dots,  A_{d})) 
                          - (1 - \delta) \det(\diag(A_{2}^{0}, A_{4}, A_{5}, \dots,  A_{d})) ]        
    \nonumber \\
    + &A_{4}^{\varepsilon} [ \det(F^{A_{2}^{0}; A_{4}}(A_{2}^{0}, A_{3}^{0}, A_{5}, \dots,  A_{d})) 
                          - (1 - \delta) \det(\diag(A_{2}^{0}, A_{3}^{0}, A_{5}, \dots,  A_{d})) ]      
    \nonumber \\
    + &\dots +     
    \nonumber \\
    + &A_{d}^{\varepsilon} [ \det(F^{A_{2}^{0}; A_{d}}(A_{2}^{0}, A_{3}^{0}, A_{4}^{0}, \dots,  A_{d-1}^{0})) 
                          - (1 - \delta) \det(\diag(A_{2}^{0}, A_{3}^{0}, A_{4}^{0}, \dots,  A_{d-1}^{0})) ] 
     \nonumber \\                     
    + &\det( P^{0}[ \{ 2, 3, \dots, d \} ] ) 
                          - (1 - \delta) \det(P^{A^{0}}[ \{ 2, 3, \dots, d \} ]) 
                 \Big\}    \,.             
  \end{align}
  Note that $F^{A_{2}^{0}; A_{d}}(A_{2}^{0}, A_{3}^{0}, A_{4}^{0}, \dots,  A_{d-1}^{0}) = P^{0}[ \{ 2, 3, \dots, d-1 \} ]$ is one of the $\binom{d}{d-2}$ principal submatrices of $F^{0}$ of size $(d-2) \times (d-2)$. The remaining $\binom{d}{d-2} - 1$ principal submatrices of size $(d-2) \times (d-2)$ can be calculated recursively in the same way as above. 
  \par
  Continuing the calculation recursively (for each of the terms) we get, 
  \begin{align}
    &\det(F) - (1 - \delta) \det(F^{A})        
    \nonumber \\
    = &\det(P^{0}[ \{ 1, 2, \dots, d \} ]) - (1 - \delta) \det(P^{A^{0}}[ \{ 1, 2, \dots, d \} ])
    \nonumber \\
        + A_{1}^{\varepsilon} \Big\{ A_{2}^{\varepsilon} \Big\{ \dots A_{d-3}^{\varepsilon} \Big\{ 
                          &A_{d-2}^{\varepsilon} [ \det(F^{A_{d-2}^{0}; A_{d-2}}(A_{d-1}, A_{d})) 
                          - (1 - \delta) \det(\diag(A_{d-1}, A_{d} )) ] 
    \nonumber \\
    + &A_{d-1}^{\varepsilon} [ \det(F^{A_{d-2}^{0}; A_{d-1}}(A_{d-2}^{0}, A_{d})) 
                          - (1 - \delta) \det(\diag(A_{d-2}^{0}, A_{d})) ]        
    \nonumber \\
    + &A_{d	}^{\varepsilon} [ \det(F^{A_{d-2}^{0}; A_{d}}(A_{d-2}^{0}, A_{d-1}^{0})) 
                          - (1 - \delta) \det(\diag(A_{d-2}^{0}, A_{d-1}^{0})) ]      
    \nonumber \\               
    + &\det( P^{0}[ \{ d-2, d-1, d \} ] ) 
                          - (1 - \delta) \det(P^{A^{0}}[ \{ d-2, d-1, d \} ]) 
                 \Big\} \dots \Big\}    
     \nonumber \\
     + & \dots               
  \end{align}
  That means, we have rewritten the term $\det(F) - (1 - \delta) \det(F^{A})$ into terms of $\binom{d}{l}$ (distinct) principal submatrices $P^{0}$ of $F^{0}$ of size $l \times l $ where $l \in \{ 3, \dots, d \}$.  
  Moreover, we are dealing with the determinants of the $2 \times 2$ matrices as follows: for example and since $A_{d} \geq A_{d}^{0}$ we have
  \begin{align}
   &\det(F^{A_{d-2}^{0}; A_{d-1}}(A_{d-2}^{0}, A_{d})) 
                          - (1 - \delta) \det(\diag(A_{d-2}^{0}, A_{d}))
     \nonumber \\
     &= \delta A_{d-2}^{0} A_{d} - | D_{d-2, d} |^{2}
     \nonumber \\
     &\geq \delta A_{d-2}^{0} A_{d}^{0} - | D_{d-2, d} |^{2}
     \nonumber \\
     &= \det( P^{0}[ \{ d-2, d \} ] ) - (1 - \delta) \det( P^{A^{0}}[ \{ d-2, d \} ] )     \,.
  \end{align} 
  The same holds analogously for the other $2 \times 2$ principal submatrices by the fact that $A_{j} \geq A_{j}^{0}$ for $j=1, \dots, d$. So, since $A_{j}^{\varepsilon} \geq 0$ for $j=1, \dots, d$ and since by assumption the inequality (\ref{eq:0 lem:the reduced $F$-property}) holds, then we can estimate
  \begin{align}
    &\det(F) - (1 - \delta) \det(F^{A})        
    \nonumber \\
    \geq &\det(P^{0}[ \{ 1, 2, \dots, d \} ]) - (1 - \delta) \det(P^{A^{0}}[ \{ 1, 2, \dots, d \} ])
    \nonumber \\
        + A_{1}^{\varepsilon} \Big\{ A_{2}^{\varepsilon} \Big\{ \dots A_{d-3}^{\varepsilon} \Big\{ 
                          &A_{d-2}^{\varepsilon} [ \det(P^{0}[ \{A_{d-1}^{0}, A_{d}^{0} \} ]) 
                          - (1 - \delta) \det(P^{A^{0}}[ \{A_{d-1}^{0}, A_{d}^{0} \} ]) ]     \allowdisplaybreaks
    \nonumber \\
    + &A_{d-1}^{\varepsilon} [ \det(P^{0}[ \{A_{d-2}^{0}, A_{d}^{0} \} ]) 
                          - (1 - \delta) \det(P^{A^{0}}[ \{A_{d-2}^{0}, A_{d}^{0} \} ]) ]       
    \nonumber \\
    + &A_{d	}^{\varepsilon} [ \det(P^{0}[ \{A_{d-2}^{0}, A_{d-1}^{0} \} ]) 
                          - (1 - \delta) \det(P^{A^{0}}[ \{A_{d-2}^{0}, A_{d-1}^{0} \} ]) ]      
    \nonumber \\               
    + &\det( P^{0}[ \{ d-2, d-1, d \} ] ) 
                          - (1 - \delta) \det(P^{A^{0}}[ \{ d-2, d-1, d \} ]) 
                 \Big\} \dots \Big\}         
     \nonumber \\
     + & \dots  
     \nonumber \\
     \geq & \, 0  .           
  \end{align}
  That means the quantity $\det(F) - (1 - \delta) \det(F^{A})$ can be estimated from below by the determinants of principal submatrices by terms as in (\ref{eq:0 lem:the reduced $F$-property}) of $F^{0}$ and so by assumption the claim follows.
  \end{proof}
  Proposition \ref{prop:the reduced $F$-property} gives us an example when the $F$-property is fulfilled. 
  \begin{prop}
   \label{prop:the reduced $F$-property}
   Assume that the covariance matrix $F_{k}^{0}$ is positive definite at all times $k = 0, 1, \dots, T$ and $S^{j}$ has independent returns for each $j = 1, \dots, d$. Then the $F$-property holds.
  \end{prop}
  \begin{proof}
  %NEW PROOF FOR THE d-DIM CASE
  Let $d \in \mathbb{N}_{\geq 2}$.
  \par
  Fix $k \in \{ 0, 1, \dots, T \}$. First we introduce the notation $\bar A_{k; j}^{0} := \Var( \rho_{k+1}^{j} )$, $\bar D_{k; i, j} := \Cov(\rho_{k+1}^{i}, \rho_{k+1}^{j})$ for $i \neq j$ where $\bar F_{k; i, j}^{0} = \bar A_{k; j}^{0}$ for $i=j$, $\bar F_{k; i, j}^{0} = \bar D_{k; i, j}$ otherwise. Our aim is to make use of Lemma \ref{lem:the reduced $F$-property}. For simplicity we omit the time $k$ and denote $F = F_{k}$.
  \par
  First note that since the covariance matrix $F^{0}$ is positive definite then
  \begin{equation}
    \det(F^{0}) >  0 \text{ and } \det(F^{A^{0}}) >  0  \,.
  \end{equation}
  Now using $\Delta S_{k+1}^{j} = S_{k}^{j} \rho_{k+1}^{j}$, the fact that $\rho_{k+1}^{j}$ is independent of $\mathcal{F}_{k}$ for all $j=1, \dots, d$, the properties of the determinant and the symmetry of the covariance matrix $F^{0}$ we get  
    \begin{align}
    \det(F^{0}) &= |S^{1}_{k}|^{2} \cdots |S^{d}_{k}|^{2} \det(\bar F^{0}) 
                        > 0   
      \nonumber \\
      \det(F^{A^{0}}) &= |S^{1}_{k}|^{2} \cdots |S^{d}_{k}|^{2} \det(\bar F^{\bar A^{0}}) 
                        > 0
  \end{align}
  with the obvious notation $\bar F_{k}^{\bar A^{0}} := diag(\bar A_{k; 1}^{0}, \dots, \bar A_{k; d}^{0})$.
  Since $S_{k}^{j} > 0$, this implies 
  \begin{align}
  \label{eq:1 the reduced $F$-property}
    \det(F^{0}) - (1 - \delta) \det(F^{A^{0}}) \geq 0
    \iff
    \det(\bar F^{0}) - (1 - \delta) \det(\bar F^{A^{0}}) \geq 0
  \end{align}
  for $\delta \in (0, 1)$.
  Furthermore, since $\bar F^{0}$ and $\bar F^{\bar A^{0}}$ are deterministic matrices with $\det(\bar F^{0}) >  0$ and $\det(\bar F^{\bar A^{0}}) > 0$, then
  \begin{align}
    \det(\bar F^{0}) - (1 - \delta) \det(\bar F^{\bar A^{0}}) 
    \geq 0 
  \end{align}
  for some $\delta \in (0, 1)$. For the $1$ principal submatrix of $F^{0}$ of size $d \times d$ which is again the matrix $F^{0}$ we want to show that 
  \begin{align}
    \det(F^{0}) + (1 - \delta) \det (F^{A^{0}})
    \geq 0   \,
  \end{align}
  which for independent returns and positive marginal price process is equivalent to $\det(\bar F^{0}) + (1 - \delta) \det (\bar F^{A^{0}}) \geq 0$ as shown in the equivalence relation (\ref{eq:1 the reduced $F$-property}). So it remains to show that for the all (distinct) $\binom{d}{l}$ principal submatrices $P^{0}$ of $F^{0}$ of size $l \times l $ where $l \in \{ 2, \dots, d - 1 \}$ we have that $\det(P^{0}) + (1 - \delta) \det (P^{A^{0}}) \geq 0$ for some $\delta \in (0, 1)$. Now using again the fact that $F_{k}^{0}$ is positive definite then we know that each principal submatrix $P^{0}$ is positive definite \citep[Observation 7.1.2]{horn.johnson:2012}. That means
  \begin{equation}
    \det(P^{0}) > 0 \text{ and } \det(P^{A^{0}}) > 0  \,.
  \end{equation} 
 Since all principal submatrices $P^{0}$ of $F^{0}$ are covariance matrices, then by the same argumentation (and obvious notation) as above we get $\det(\bar P^{0}) - (1 - \delta) \det(\bar P^{\bar A^{0}}) \geq 0$ for some $\delta \in (0, 1)$ which for independent returns and $S_{k}^{j} > 0$ is equivalent to 
  \begin{align}
    \det( P^{0}) + (1 - \delta) \det ( P^{A^{0}}) \geq 0    \,.
  \end{align}
  Finally, from Lemma \ref{lem:the reduced $F$-property} the claim follows.  
   \end{proof}
    \begin{prop}
   \label{prop:the reduced $F$-property for independent increments}
    Assume that the covariance matrix $F_{k}^{0}$ at all times $k = 0, 1, \dots, T$ is positive definite and $S^{j}$ has independent increments for each $j = 1, \dots, d$. Then the $F$-property holds.
  \end{prop}
  \begin{proof}
  Follows by analogous arguments as in Proposition \ref{prop:the reduced $F$-property}.
   \end{proof}
   \begin{remark}
     Note that rewriting Lemma \ref{lem:the reduced $F$-property} when $\varepsilon = 0$ then the condition simply reduces to the covariance matrix being such
    \begin{equation}
      \det(F^{0}) - (1 - \delta) \det(F^{A^{0}})
      \geq 0
    \end{equation}
    for some $\delta \in (0, 1)$ and principal submatrices do not need to be considered. 
   \end{remark}
  \begin{remark}
  \label{rem:the 2 dim F property}
    In the $2$-dimensional case in order to ensure that $F_{k}^{0}$ is positive definite\footnote{Recall that a matrix $F$ is positive definite if and only if its leading principal minors are all positive.}, i.e. $A_{k; 1}^{0} A_{k; 2}^{0} - D_{k; 1, 2}^{2} > 0$, $A_{k; 1}^{0} > 0$, $A_{k; 2}^{0} > 0$, in the case of independent returns (or increments) we just need strict Cauchy-Schwarz inequality, which means that $S^{1}$ and $S^{2}$ must be linearly independent. Then Proposition \ref{prop:the reduced $F$-property} can be applied. 
  \end{remark}

   %Nonnegative supply curve
  \subsection{Nonnegative supply curve}
  \label{sec:nonnegative supply curve}
In this section we consider the $1$-dimensional case for simplicity. An extension to the multidimensional case is straightforward. As we already mentioned the (linear) supply curve  $S_{k}(x) = (1 + x \varepsilon_{k}) S_{k}$ can also take negative values when a negative transaction $x$ is such that $x\leq -1/\varepsilon$. So, a natural question to ask is how one could define a function $h: \mathbb{R} \to \mathbb{R}$ so that the supply curve process %%% \footnote{Note that in this case of a market sell, the process $S_{k}$ represents the best bid price at time $k$.}
  \begin{equation}
    S_{k}(x) = h(x)S_{k}
  \end{equation}
  is nonnegative. This can be done for example by the function
  \begin{equation}
    h(x)
      = ( 1 + x \varepsilon_{k} ) \mathbf{1}_{ \{ x \geq -z_{k} \} } 
      + ( 1 - z_{k} \varepsilon_{k} ) \mathbf{1}_{ \{ x < -z_{k} \} } 
  \end{equation}
 defined for some deterministic positive process $z = (z_{k})_{k=0, 1, \dots, T}$ where $0 < z_{k} \leq 1/\varepsilon_{k}$ for all $k=0, 1, \dots, T$. Then $z_{k} S_k$ represents a lower bound for the price received when selling a large quantity of shares.
   \par
  The corresponding cost process under illiquidity $ \hat C^{b}(\varphi) = (\hat C_{k}^{b}(\varphi))_{k=0,1, \dots, T}$ of a strategy $\varphi = (X, Y)$ is then
   \begin{align}
       \hat C_{k}^{b}(\varphi) 
         := V_{k} (\varphi) - \sum_{m=1}^{k} X_{m} \Delta S_{m} 
         & + \sum_{m=1}^{k} \varepsilon_{m} S_{m} | \Delta X_{m+1} |^{2} \mathbf{1}_{ \{ \Delta X_{m+1} \geq -z_{m} \} }
         \nonumber \\ &
          - \sum_{m=1}^{k} z_{m} \varepsilon_{m} S_{m} \Delta X_{m+1} \mathbf{1}_{ \{ \Delta X_{m+1} < -z_{m} \} }.
   \end{align} 
 Moreover, as in Section \ref{sec:existence and recursive construction of an optimal strategy under illiquidity} and by Proposition \ref{prop:minimizing variance} at time $k$ we want to minimize the expression (w.l.o.g. $\alpha = 1$)
    \begin{align}
       & \condVar{ V_{k+1}(\varphi) - X'_{k+1} \Delta S_{k+1} }
       \nonumber \\& \quad \quad \quad \quad 
       + \condE{ \varepsilon_{k+1} S_{k+1} | X_{k+2} - X'_{k+1} |^{2} \mathbf{1}_{ \{ X_{k+2} - X'_{k+1} \geq -z_{k+1} \} } }
       \nonumber \\& \quad \quad \quad \quad 
       - \condE{ z_{k+1} \varepsilon_{k+1} S_{k+1} (X_{k+2} - X'_{k+1}) \mathbf{1}_{ \{ X_{k+2} - X'_{k+1} < -z_{k+1} \} } }
     \end{align}
    over all appropriate $X'_{k+1}$. Rewriting the above expression, one needs to minimize the function $\hat f_{k}^{b} : \mathbb{R} \times \Omega \rightarrow \mathbb{R}^{+}$ defined by
    \begin{align}
      \hat f_{k}^{b}(c, \omega)
        &= | c |^{2} \hat A_{k}^{b}(\omega)
          - 2 c \hat b_{k}^{b}(\omega)
          + c \hat d_{k}^{b}(\omega)
        \nonumber \\
        &+ \condVar{ V_{k+1} }(\omega) 
          + \condE{ \varepsilon_{k+1} S_{k+1} | X_{k+2} |^{2} \mathbf{1}_{ \{ X_{k+2} - c \geq -z_{k+1} \} } }(\omega) 
          \nonumber \\ 
          & \quad \quad \quad \quad \quad \quad \quad 
          - \condE{ z_{k+1} \varepsilon_{k+1} S_{k+1} X_{k+2} \mathbf{1}_{ \{ X_{k+2} - c < -z_{k+1} \} } } (\omega) 
    \end{align}
  where the following notation is used, 
  \begin{align}
    \hat A_{k}^{b}
    &= \condVar{ \Delta S_{k+1} }
      + \condE{ \varepsilon_{k+1} S_{k+1} \mathbf{1}_{ \{ X_{k+2} - c \geq -z_{k+1} \} } }
    \nonumber \\
    \hat b_{k}^{b}
    &= \condCov{ V_{k+1}, \Delta S_{k+1} }
      + \condE{ \varepsilon_{k+1} S_{k+1} X_{k+2} \mathbf{1}_{ \{ X_{k+2} - c \geq -z_{k+1} \} } }
    \nonumber \\
    \hat d_{k}^{b}
    &= \condE{ z_{k+1} \varepsilon_{k+1} S_{k+1} \mathbf{1}_{ \{ X_{k+2} - c < -z_{k+1} \} } }   \,.
  \end{align}
  Furthermore, under similar arguments and assumptions as in Sections \ref{sec:existence and recursive construction of an optimal strategy under illiquidity} and \ref{sec:assumptions}, one can use the dominated convergence theorem to show that the equation $\frac{d}{dc} \hat f_{k}^{b}(c) = 0$ gives that the optimal strategy $\hat{\varphi}=(\hat{X}, \hat{Y})$ fulfills the implicit relation
  \begin{align}
    \hat X_{k+1}
    = \frac{ \condCov{ V_{k+1}, \Delta S_{k+1} }
             + \condE{ \varepsilon_{k+1} S_{k+1} \hat X_{k+2} \mathbf{1}_{ \{ \hat X_{k+2} - \hat X_{k+1} \geq -z_{k+1} \} } } 
             - \frac{1}{2} Q  
               }
              { \condVar{ \Delta S_{k+1} }
              + \condE{ \varepsilon_{k+1} S_{k+1} \mathbf{1}_{ \{ \hat X_{k+2} - \hat X_{k+1} \geq - z_{k+1} \} } }
               }
  \end{align} 
  with 
  \begin{equation}
    Q = \condE{ z_{k+1} \varepsilon_{k+1} S_{k+1} \mathbf{1}_{ \{ \hat X_{k+2} - \hat X_{k+1} < -z_{k+1} \} } }.
  \end{equation}

\section{Application to Electricity Markets}
\label{sec:applications}
  
In this section we apply the previous results to hedge an Asian-style electricity option with electricity futures that are exposed to liquidity costs. These futures might have different maturities, i.e.~certain hedge instruments might terminate before maturity of the option (final time horizon $T$) and hedging in these instruments is only possible on certain subintervals of $[0,T]$. A priori this situation is not covered by our setting in the previous sections where it is assumed that hedging is possible until $T$ in all hedge instruments. In Subsection~\ref{sec:the general case}, we thus shortly sketch how hedge instruments with different maturities can be embedded in our setting from the previous sections, before we focus our example on electricity markets in Subsection~\ref{sec:the energy market case}.

   %The General Case
   \subsection{Hedge instruments with different maturities}
   \label{sec:the general case}

\par
On our stochastic basis $(\Omega, \mathcal{F}, \mathbb{F}, \mathbb{P})$ with final time horizon $T$, consider now nonnegative price processes $S^j = (S^j_{k})_{k=0, 1, \dots, T_j}$ of $d$ available hedge instruments with maturity $T_j\le T$, $j=1, \dots, d$. That is, hedging in asset $j$ is only possible until time $T_j\le T$, $j=1, \dots, d$, where without loss of generality we assume $0 < T_{1} \le T_{2} \le \dots \le T_{d}\le T$. To fit this situation into our general setting, we introduce an associated $d$-dimensional price process  $\tilde S = (\tilde S_{k})_{k=0, 1, \dots, T}$ by artificially keeping each asset $S^j$ constant on the remaining interval $[T_j,T]$:
\begin{equation}
\label{eq:price process general applications}
  \tilde S_{k}^{j}
    = S_{k}^{j} \mathbf{1}_{[0, T_{j} )}(k) + S_{T_{j}}^{j} \mathbf{1}_{[T_{j}, T ]}(k)
\end{equation}   
for $j=1, \dots, d$ and $k \in \{ 0, 1, \dots, T \}$. Moreover, we consider a positive, deterministic $\mathbb{R}^{d}$-valued liquidity process $\varepsilon=(\varepsilon_{k})_{k=0, 1, \dots, T}$, which is extended by some $\varepsilon_{m}^{j} > 0$ on the intervals $m \in[T_{j},T] $ for all $j \in \{ 1, \dots, d \}$, i.e.~we assume positive liquidity costs during the extended price dynamics. 

It is then clear already intuitively that an investor would not trade in asset $j$ during the interval $[T_{j},T]$ since during this time frame the asset generates zero gains while incurring positive liquidity costs. Indeed, employing the fact for $k \geq T_{l}$ we have $\Delta \tilde S_{k+1}^{l} = 0$, it is straightforward to see from Proposition \ref{prop:minimizing variance}, Property \ref{prop:2}, that in this situation a LRM-strategy must be of the form $\tilde X_{m}^{l} = 0 \text{ for } m = T_{l}+1, \dots, T \,\,, l \in \{ 1, \dots, d \}$. i.e.~the hedger liquidates his position in the $j$-th asset at time $T_{j}+1$. Thus, in our extended market a LRM-strategy $\tilde X$ automatically respects the original hedge constraints beyond maturities $T_j$, $j \in \{ 1, \dots, d \}$ and is thus also a LRM-strategy in our setting with hedge instruments with different maturities. In the following we say the asset $\tilde S^j$ is \emph{active} at time $k$ if $k\le T_j$ and \emph{in
 active} at time $k$ if $k > T_j$.

The existence and computation of a LRM-strategy under a linear supply curve $\tilde S_{k}^{j}(x^{j}) = \tilde S_{k}^{j} + x^{j}\varepsilon_{k}^{j} \tilde S_{k}^{j}$ as developed in Section~\ref{sec:linear:supply} now takes the following form for hedge instruments with different maturities. Using the fact that a LRM-strategy $\tilde X$ fulfills $\tilde X_{m}^{l} = 0 \text{ for } m = T_{l}+1, \dots, T \,\,, l \in \{ 1, \dots, d \}$, the minimization at step $k \in \{ 0, 1, \dots, T-1 \}$ of the function $f_k$ in \eqref{eq:minimizing function} reduces to the minimization of the function $\tilde f_{k} : \mathbb{R}^{d-l_k} \times \Omega \rightarrow \mathbb{R}^{+}$ defined by 
    \begin{align}\label{ftilde}
      \tilde f_{k}(c, \omega)
        &= \sum_{j = l_k + 1}^{d} | c_{j} |^{2} A_{k; j}(\omega)
          - 2 \sum_{j = l_k + 1}^{d} c_{j} b_{k; j}(\omega)
          + \sum_{j \neq i, j, i \geq l_k + 1}^{d} c_{j} c_{i} D_{k; j, i}(\omega)
        \\
        &+ \condVar{ V_{k+1} }(\omega) 
          + \sum_{j = l_k + 1}^{d} \condE{
               \varepsilon_{k+1}^{j} S_{k+1}^{j} | X_{k+2}^{j} |^{2} }(\omega)   \nonumber 
    \end{align}
where the sums are only over the assets $\tilde S^j$, $j=l+1,...,d$, that are active during the $k$'th period, i.e.~$l_k:=\max\{r\in \{ 1, \dots, d \}:T_{r} < k\}$. Thus, the conditions required in Theorem~\ref{theo:existence} for existence of a LRM-strategy reduce to lower-dimensional conditions that in each period only concern the active hedge instruments. More precisely, using the notation from Section \ref{sec:linear:supply}, we define for each period $k \in \{ 0, 1, \dots, T-1 \}$ the symmetric matrix $\tilde F_{k} \in \mathbb{R}^{d-l_k \times d-l_k}$ (a principal submatrix of $F_{k}$) by $\tilde F_{k; i, j} = D_{k; i+l_k, j+l_k}$ for $i \neq j$, $\tilde F_{k; i, j} = A_{k; j+l_k}$ for $i = j$, $i, j \in \{ 1, \dots, d-l_k \}$ and $\tilde b_{k} := ( b_{k; l_k+1}, \dots, b_{k; d} )^{*} \in \mathbb{R}^{d-l_k}$. Then minimizing \eqref{ftilde} amounts to solving the linear system
  \begin{equation}
    \tilde F_{k} c = \tilde b_{k}   \,.
  \end{equation} 
in $c\in \mathbb{R}^{d-l_k}$. Note that $\tilde F_{k} = \tilde F_{k}^{0} + \tilde F_{k}^{\varepsilon}$ where $\tilde F^{\varepsilon} = diag(A_{k;l+1}^{\varepsilon}, \dots, A_{k;d}^{\varepsilon})$ and $\tilde F_{k}^{0}$ is the matrix $\tilde F_{k}$ with $\varepsilon_{k+1}^{j} = 0$ for $j=l+1, \dots, d$, that is a reduced form of the covariance matrix of the price process $\tilde S$. Following the arguments in Section \ref{sec:linear:supply}, we then get the following version of Theorem~\ref{theo:existence} on the existence of a LRM-strategy in the context of hedge instruments with different maturities:

\begin{cor}
  \label{cor:existence}
    Consider a contingent claim $H = \bar X_{T+1}^{*} S_{T} + \bar Y_{T} \in \mathbb{L}_{T}^{2, 1}$ with $\bar X_{T+1} = 0$ and a price process of the form in equation (\ref{eq:price process general applications}). Assume that for each $k$-th period, the covariance matrix $\tilde F_{k}^{0}$ is positive definite. Furthermore assume that bounded mean-variance tradeoff, the $F$-property and the $F$-diagonal condition hold for the active assets in the $k$-th period at time $k \in \{ 0, 1, \dots, T-1 \}$. Then there exists a LRM-strategy $\hat \varphi = ( \hat X, \hat Y )$ under illiquidity with $\hat X_{T+1} = 0$, $\hat Y_{T} = H$. In particular for $k \in \{ 0, 1, \dots, T-1 \}$ we have $\hat X = (\bar 0, \tilde X)$ with $\bar 0 = (0, \dots, 0) \in \mathbb{R}^{l_{k}}$ and
    \begin{equation}
      \tilde X_{k+1} 
      = \tilde F_{k}^{-1} \tilde b_{k}
      \quad \mathbb{P}-a.s.
    \end{equation} 
    in $\mathbb{R}^{d-l_{k}}$ and for $k \in \{ 0, 1, \dots, T-1 \}$ 
    \begin{equation}
      \hat Y_{k}
      = \condE{ \hat W_{k} } - \hat X_{k+1}^{*} \tilde S_{k}
      \quad \Pas
    \end{equation}
    where $\hat W_{k} = H - \sum_{m=k+1}^{T} \hat X_{m}^{*} \Delta \tilde S_{m}$.
  \end{cor}
 
% \begin{theo}
  %  \label{theo:existence}
  %  Assume that $S$ has the $F$-property, bounded mean-variance tradeoff and satisfies the $F$-diagonal condition. Let further the covariance matrix $F_{k}^{0}$ be positive definite at all times $k = 0, 1, \dots, T$. Then for any contingent claim $H = \bar X_{T+1}^{*} S_{T} + \bar Y_{T} \in \mathbb{L}_{T}^{2, 1}$ with $\bar X_{T+1}^{*} S_{T} \in \mathbb{L}_{T}^{2, 1}$ and $\bar X_{T+1} \in \mathbb{L}_{T}^{2, d}$, there exists a local risk-minimizing strategy $\hat \varphi = ( \hat X, \hat Y )$ under illiquidity with $\hat X_{T+1} = \bar X_{T+1}$ and $\hat Y_{T} = \bar Y_{T}$. Furthermore, the strategy has the representation
  %  \begin{align}
    %  \hat X_{k+1} 
   %   &= F_{k}^{-1} b_{k}
    %  \quad \Pas \text{ for } k = 0, \dots, T-1
    %  \label{eq:LRM-strategy X}
    %  \\
    %  \hat Y_{k}
    %  &= \condE{ \hat W_{k} } - \hat X_{k+1}^{*} S_{k}
    %  \quad \Pas \text{ for } k = 0, 1, \dots, T-1
     % \label{eq:LRM-strategy Y}
  %  \end{align}
  %  where $\hat W_{k} = H - \sum_{m=k+1}^{T} \hat X_{m}^{*} \Delta S_{m}$.
 % \end{theo}

  %The Energy Market Case
   \subsection{LRM strategies in electricity markets}
   \label{sec:the energy market case}
In the remaining parts of the section, we now consider the example of hedging an Asian-style electricity option with electricity futures under liquidity costs by a LRM-strategy. The price processes for electricity futures we are considering are based on a continuous-time multi-factor spot price model proposed in~\cite{benth.meyerbrandis.kallsen:2007}, which we recall in Subsection~\ref{model} before we explicitly compute and simulate LRM-strategies in an example in Subsection~\ref{simulation}. 

\subsubsection{An electricity market model} \label{model} 

In~\cite{benth.meyerbrandis.kallsen:2007}, the price $E(t)$ of spot electricity at time $t\in[0,T]$ is modeled by 
  \begin{equation}\label{spotmodel}
     E(t)
     =\sum_{i=1}^{n}  \Lambda_i(t) Y_{i}(t)\,,
   \end{equation}
 where for $i=1, \dots, n$ the positive and deterministic function $\Lambda_i$ accounts for seasonality and $Y_{i}$ is the solution to an Ornstein-Uhlenbeck stochastic differential equation
   \begin{equation}
     dY_{i}(t)
     = - \lambda_{i} Y_{i}(t) dt + \sigma_{i}(t) dL_{i}(t)
     \quad , \quad 
     Y_{i}(0) = y_{i}\,,
   \end{equation}
 where $\lambda_{i} > 0$ are constants, and $\sigma_{i}(t)$ are deterministic, positive bounded functions.
   Moreover, the $L_{i}$'s are independent, increasing pure jump L\'{e}vy processes with jump measures $N_{i}(dt, dz)$ which have deterministic predictable compensators of the form \linebreak $\nu_{i}(dt, dz) = dt \nu_{i}(dz)$. Note that by the increasing nature of the $L_{i}$'s the positivity of the $Y_{i}$'s and thus also of the spot price $E$ is ensured.  We assume that the model~\eqref{spotmodel} is defined on a stochastic basis $( \Omega, \mathbb{F}, (\mathcal{F}_{t})_{0 \leq t \leq T}, \mathbb{P} )$ where the filtration $(\mathcal{F}_{t})_{0 \leq t \leq T}$ is generated by the $L_{i}$'s.
   
The available hedge instruments are electricity futures, which, by the flow character of electricity, delivers spot electricity over a delivery period $[ T_{1}^{F}, T_{2}^{F} ]$ for $T_{1}^{F}<T_{2}^{F}\le T$ rather than at a fixed point in time. That is, the pay-off of the (financially settled) futures at the end of the delivery period is 
   \begin{equation}
   \frac{1}{T_{2}^{F} - T_{1}^{F}} \int_{T_{1}^{F}}^{T_{2}^{F}} E(u) du  \,,
   \end{equation}
 and the life of the asset terminates at $T_{2}^{F}$. In order to compute the price dynamics of an electricity futures we assume for simplicity that $\mathbb{P}$ is already an equivalent martingale measure, such that the the price $F(t;T_{1}^{F}, T_{2}^{F})$ of the futures at time $t\le T_{2}^{F}$ as a traded asset is given by 
%   
%   \par
%   The processes $L_{i}$'s can be represented in terms of the jump measures as follows, 
%   \begin{equation*}
%     L_{i}(t)
%     = \int_{0}^{t} \int_{0}^{\infty} z N_{i}(ds, dz)
%   \end{equation*} 
%   and setting $\tilde N_{i}(ds, dz) = N_{i}(ds, dz) - \nu_{i}(ds, dz)$, the compensated jump process $\tilde L_{i}$ can be written as 
%   \begin{equation*}
%     \tilde L_{i}(t)
%     = \int_{0}^{t} \int_{0}^{\infty} z \tilde N_{i}(ds, dz)      \,.
%   \end{equation*} 
   \begin{equation} \label{futures}
     F(t;T_{1}^{F}, T_{2}^{F})
     = \E \left[ \frac{1}{T_{2}^{F} - T_{1}^{F}} \int_{T_{1}^{F}}^{T_{2}^{F}} E(u) du \middle| \mathcal{F}_{t} \right ]    \,.
   \end{equation}
 Using the explicit solution
      \begin{equation}
         Y_{i}(u) 
         = Y_{i}(t) \e^{-\lambda_{i}(u-t)}
          + \int_{t}^{u} \sigma_{i}(s) \e^{-\lambda_{i}(u-s)} dL_{i}(s)
       \end{equation}
 for the Ornstein-Uhlenbeck components $Y_i$, $i=1,...n$, a straightforward computation of the conditional expectation in \eqref{futures} yields the following price of futures contracts in the continuous-time spot model :
    \begin{prop}
     \label{prop:future dynamics}
     The price $F(t, T_{1}^{F}, T_{2}^{F})$ at time $t$ of an electricity futures with delivery period $[T_{1}^{F}, T_{2}^{F}]$ is given by
 \begin{align}
       F(t, T_{1}^{F}, T_{2}^{F})
       &= \sum_{i=1}^{n} Y_{i}(t) \frac{1}{T_{2}^{F} - T_{1}^{F}} 
                \int_{T_{1}^{F}}^{T_{2}^{F}} \Lambda_{i}(u) \e^{-\lambda_{i}(u-t)} du
       \nonumber \\
       &+ \frac{1}{T_{2}^{F} - T_{1}^{F}} \int_{T_{1}^{F}}^{T_{2}^{F}} \int_{t}^{u} \int_{\mathbb{R}^{+}} 
                \sigma_{i}(s) \Lambda_{i}(u) \e^{-\lambda_{i}(u-s)} z  \nu_{i}(dz) ds du
     \end{align}
   for $0 \leq t \leq T_{1}^{F}$, and
       \begin{align}
       F(t, T_{1}^{F}, T_{2}^{F})
       = \frac{1}{T_{2}^{F} - T_{1}^{F}} \int_{T_{1}^{F}}^{t} E(u) du
       &+ \sum_{i=1}^{n} Y_{i}(t) \frac{1}{T_{2}^{F} - T_{1}^{F}} 
                \int_{t}^{T_{2}^{F}} \Lambda_{i}(u) \e^{-\lambda_{i}(u-t)} du
       \nonumber \\
       &+ \frac{1}{T_{2}^{F} - T_{1}^{F}} \int_{t}^{T_{2}^{F}} \int_{t}^{u} \int_{\mathbb{R}^{+}} 
                \sigma_{i}(s) \Lambda_{i}(u) \e^{-\lambda_{i}(u-s)} z  \nu_{i}(dz) ds du
     \end{align}
   for $T_{1}^{F} \leq t \leq T_{2}^{F}$.
   \end{prop}
   
Based on this continuous-time spot and futures price model, we now construct a discrete-time electricity market model that fits into our framework by sampling the continuous-time processes at finitely many trading times $0=t_0, t_1,...,T$, i.e.~our hedge instruments $S^j$, $j=1,...d$, are given by futures price processes of the form
\begin{equation}
S^j_{k}:=F^j(t_k, T_{1}^{F^j}, T_{2}^{F^j}) \quad \text{for}\quad 0\le t_k\le T_{2}^{F^j}\le T \,.
\end{equation}
In the following we always assume that delivery period times are part of the discrete time grid, i.e.~$T_{1}^{F}, T_{2}^{F}\in \{t_0, t_1,...,T\}$. After the maturity $T_{2}^{F}$, the futures contract ceases to exist and trading is not possible anymore. During the delivery period $[T_{1}^{F},T_{2}^{F}]$, depending on the conventions of the exchange, trading is either not possible at all or very illiquid. We capture this feature by specifying high liquidity costs during $[T_{1}^{F},T_{2}^{F}]$, with the impossibility of trading as the limit case when liquidity costs tend to infinity. Before the delivery period, one typically observes on electricity markets that futures become the more liquid the shorter the remaining time to delivery period is. We capture this behavior by the following liquidity structure $\varepsilon^{j}$ for the futures $F^j$, $j=1,...d$: 
%\oldver{which is exponentially decreasing in time until the start of the delivery period and then jumps to a constant (high) level during the delivery period:}

\begin{align} \label{TimeLiquidity}
  \varepsilon^{j}_{t}
  &= a_{j} (1 - \exp(-(T_{1}^{F^{j}} - t))) + \delta_j
  \quad , \quad
  a_{j}
  = M_{j} \frac{1}{1 - \exp(-T_{1}^{F^{j}})}
  \quad \text{ for } 0 \leq t \leq T_{1}^{F^{j}}\,,
  \nonumber \\
  \varepsilon^{j}_{t}
  &= N_{j}
  \quad \text{ for } T_{1}^{F^{j}} < t \leq T_{2}^{F^{j}}\,. 
\end{align}
The liquidity structure $\varepsilon^{j}$ for a future $F^j$ thus starts from a constant $M_j>0$ at time $0$ and decreases exponentially in time until the start of the delivery period to a level $\delta_j>0$.  During the delivery period it then jumps to a constant (high) level $N_j>0$.

Further, in our simulation study we compare the time varying liquidity structure in \eqref{TimeLiquidity} with a constant liquidity structure given by
\begin{align} \label{ConstantLiquidity}
  \varepsilon^{j}_{t}
  = M_{j} 
  \quad \text{ for } 0 \leq t \leq T_{1}^{F^{j}}\,,
   \quad \quad
  \varepsilon^{j}_{t}
  = N_{j}
  \quad \text{ for } T_{1}^{F^{j}} < t \leq T_{2}^{F^{j}}\,. 
\end{align}
for $M_j>0$ and $N_j>0$.
   
\subsubsection{LRM-strategies of electricity call options}   \label{simulation}

In the electricity market model specified in Subsection~\ref{model}, we now intend to compute a LRM-strategy of a financially settled Asian call option written on an electricity future with delivery period $[T_{1}^{c},T_{2}^{c}]$ for $0<T_{1}^{c}< T_{2}^{c}\le T$, i.e.~the claim is given by $H = \bar Y_{T}$ with  
   \begin{equation}\label{calloption}
    \bar Y_{T} = \left( \frac{1}{T_{2}^{c} - T_{1}^{c}} \int_{T_{1}^{c}}^{T_{2}^{c}} E(u) du - K \right)^{+}
   \end{equation}
for some strike price $K$. In the following we will always assume that the option maturity is equal to the terminal time horizon: $T_{2}^{c} = T$. 

 \par
We will analyze and compare various specifications where the investor can hedge in two different futures $F^{1}$, $F^{2}$ with corresponding delivery periods $[ T_{1}^{F^{1}}, T_{2}^{F^{1}} ]$ and  $[ T_{1}^{F^{2}}, T_{2}^{F^{2}} ]$, respectively, where we assume $T_{2}^{F^{1}} \leq T_{2}^{F^{2}} \leq T$ and $T_{1}^{F^{1}} \neq T_{1}^{F^{2}}$.\footnote{Basically one needs that either $T_{2}^{F^{1}} \neq T_{2}^{F^{2}}$ or $T_{1}^{F^{1}} \neq T_{1}^{F^{2}}$ so that the conditional Cauchy-Schwarz inequality is strict. See Remark \ref{rem:the 2 dim F property}.} In this situation, Corollary~\ref{cor:existence} ensures the existence of a LRM-strategy under liquidity costs. Indeed,  from Proposition~\ref{prop:independent increments} it is clear that both the bounded mean-variance tradeoff and the $F$-diagonal condition hold for the active assets in each period by the fact that the futures have independent increments. Moreover, by Proposition \ref{prop:the reduced $F$-property for i
 ndependent increments} and Remark \ref{rem:the 2 dim F property}, it remains to check if the conditional Cauchy-Schwarz-Inequality is strict, i.e. if for each $k\in\{0,...,T_{2}^{F^{1}}\}$ the active hedge instruments $F^{1}$ and $F^{2}$ fulfill
   \begin{equation}
     \condCov{ \Delta F_{k+1}^{1}, \Delta F_{k+1}^{2} }^{2}
     <
     \condVar{ \Delta F_{k+1}^{1} }
     \condVar{ \Delta F_{k+1}^{2} }\,,
   \end{equation}
which ensures that the inverse matrix $\tilde F_{k}^{-1}$ exists and additionally the $F$-property holds. The CS-inequality is indeed strict since $T_{1}^{F^{1}} \neq T_{1}^{F^{2}}$ and this ensures that $\P(F_{k+1}^{1} = a F_{k+1}^{2}) < 1$ for any constant $a \in \mathbb{R}$.\footnote{That means, both futures are linearly independent with positive probability.} So by Corollary~\ref{cor:existence}, there exists a LRM-strategy $\hat \varphi = ( \hat X, \hat Y )$ under illiquidity of the form $\hat X_{T+1} = 0$, $\hat Y_{T} = H$ and  $\hat X = (\bar 0, \tilde X)$ with $\bar 0 = (0, \dots, 0) \in \mathbb{R}^{l_{k}}$ and 
    \begin{equation}
      \tilde X_{k+1} 
      = \tilde F_{k}^{-1} \tilde b_{k}
      \quad \Pas
    \end{equation} 
in $\mathbb{R}^{d-l_{k}}$ for $k \in \{ 0, \dots, T-1 \}$. Note that the matrix $\tilde F_{k}^{-1}$ is $2 \times 2$-dimensional for $k \in \{ 0, \dots, T_{2}^{F^{1}} - 1 \}$ and $1$-dimensional for $k \in \{ T_{2}^{F^{1}}, \dots, T_{2}^{F^{2}} - 1 \}$.
\\ \\
To compute the optimal strategy $\tilde X$ one needs to compute conditional expectations of the form $\E[ Y | X ]$ for square integrable random variables $X$ and $Y$. A popular method to compute such conditional expectations numerically, which we also employ in the following, is the \textit{least-squares Monte Carlo} (LSMC) method first used in finance by \cite{longstaff.schwartz:2001} for the valuation of American options. We do not go into further details of the LSMC method, but just mention that we use indicator functions constructed via the \textit{binning} method as basis functions. We refer to \cite{fries:2007} for a nice introduction to the LSMC method.
\par
In our $2$-dimensional example we need to simulate,
\begin{align}
  \tilde X_{T+1} 
  &= 0
  \nonumber \\
  \tilde X_{k+1}
  &= \frac{1}{A_{k; 2}} b_{k; 2}
       \quad \text{ for } k \in \{ T_{2}^{F^{1}}, \dots, T_{2}^{F^{2}} - 1  \}
  \nonumber \\%
  \tilde X_{k+1}
  &= (\tilde X_{k+1}^{1}, \tilde X_{k+1}^{2})
  \quad \text{ for } k \in \{ 0, \dots, T_{2}^{F^{1}} - 1 \}, \quad \text{ where }
  \nonumber \\%
  \tilde X_{k+1}^{1}
  &= \frac{1}{A_{k; 1}A_{k; 2} - | D_{k; 1, 2} |^{2} } ( A_{k; 2} b_{k; 1} - D_{k; 1, 2} b_{k; 2} )
  \nonumber \\
  \tilde X_{k+1}^{2}
  &= \frac{1}{A_{k; 1}A_{k; 2} - | D_{k; 1, 2} |^{2} } ( A_{k; 1} b_{k; 2} - D_{k; 1, 2} b_{k; 1} )\,.
\end{align} 
To implement the LSMC-method one needs to ensure that all random variables in the conditional expectations are square integrable. This is guaranteed by  Corollary~\ref{cor:LSMC integrability condition} below, which is mostly based on Lemma \ref{lem:the alpha, beta, gamma terms}. 
For Corollary \ref{cor:LSMC integrability condition}, we use the notation of Section~\ref{sec:the general case} where $\tilde S=(\tilde S^{1}, \dots, \tilde S^{d})$ is the price process of the (extended) hedge instruments.

\begin{cor}
  \label{cor:LSMC integrability condition}
  Assume that the components of the marginal price process $\tilde S$ and the contingent claim $H$ are both in $\classlfourone$ as well as $\bar X_{T+1}=0$. Under the assumptions of Corollary \ref{cor:existence} there exists a LRM-strategy $\hat \varphi = ( \hat X, \hat Y )$ under illiquidity such that for some constant $C > 0$
  \begin{align}
       \E[ ( (\tilde F_{k}^{-1} \tilde b_{k})_{j} \Delta \tilde S_{k+1}^{j} )^{4} ]
       &\leq C ( \E | V_{k+1}(\hat \varphi) |^{4}
        + \sum_{i=1}^{d}\E | \hat X_{k+2}^{i} |^{4} )
       \\
       \E[ ( (\tilde F_{k}^{-1} \tilde b_{k})_{j} )^{4} ]
       &\leq C ( \E | V_{k+1}(\hat \varphi) |^{4}
        + \sum_{i=1}^{d}\E | \hat X_{k+2}^{i} |^{4} )     
     \end{align}
    for $k \in \{ 0, 1, \dots, T - 1 \}$ where $V_{k+1}(\hat \varphi) = \condEplus{ H - \sum_{m=k+2}^{T} \hat X_{m}^{*} \Delta \tilde S_{m} }$. In particular, all random variables in the conditional expectations in the terms $A_{k, j}$, $b_{k, j}$ and $D_{k, j, i}$ are square integrable for all $j=l_{k}+1, \dots, d$ and $k= 0, 1, \dots, T-1$.
  \end{cor}
  \begin{proof}
  The existence of a LRM-strategy $\hat \varphi = ( \hat X, \hat Y )$ under illiquidity follows directly from Corollary \ref{cor:existence}. The fact that $V_{k+1}(\hat \varphi) = \condEplus{ H - \sum_{m=k+2}^{T} \hat X_{m}^{*} \Delta \tilde S_{m} }$ follows also directly from $\hat Y_{k}$ defined as in Corollary \ref{cor:existence}.
  \par
  By Lemma \ref{lem:the alpha, beta, gamma terms} together with Lemma \ref{lem:alpha beta terms boundedness} applied for the active assets at time $k \in \{ 0, 1, \dots, T-1 \}$, we get
  \begin{align}
       \E[ ( (\tilde F_{k}^{-1} \tilde b_{k})_{j} \Delta \tilde S_{k+1}^{j} )^{4} ]
       \leq C \E[ ( \condVar{ V_{k+1}(\hat \varphi) }
        + \sum_{i=1}^{d} \condE{ | \hat X_{k+2}^{i} |^{2} } )^{2} ]   \,.
   \end{align}
   Furthermore, using $\Var [X] \leq \E[X^2]$ we can estimate, 
   \begin{align}
       \E[ ( (\tilde F_{k}^{-1} \tilde b_{k})_{j} \Delta \tilde S_{k+1}^{j} )^{4} ]
       &\leq C \E[ ( \condE{ | V_{k+1}(\hat \varphi) |^{2} }
        + \sum_{i=1}^{d} \condE{ | \hat X_{k+2}^{i} |^{2} } )^{2} ]   
      \nonumber \\
      & \leq C \E[  \condE{ | V_{k+1}(\hat \varphi) |^{4} }
        + \sum_{i=1}^{d} \condE{ | \hat X_{k+2}^{i} |^{4} }  ] 
      \nonumber \\
      & = C ( \E | V_{k+1}(\hat \varphi) |^{4} 
        + \sum_{i=1}^{d} \E | \hat X_{k+2}^{i} |^{4} )
   \end{align}
   where for the last inequality we have used the conditional Jensen Inequality and for the equality we have applied the tower property. Analogously we also get the second inequality of the claim.
   \par
    This shows that,
   \begin{align}
       \E[ ( \hat X_{k+1}^{j} \Delta \tilde S_{k+1}^{j} )^{4} ]
       &\leq C ( \E | V_{k+1}(\hat \varphi) |^{4}
        + \sum_{i=1}^{d}\E | \hat X_{k+2}^{i} |^{4} )
       \\
       \E[ ( \hat X_{k+1}^{j} )^{4} ]
       &\leq C ( \E | V_{k+1}(\hat \varphi) |^{4}
        + \sum_{i=1}^{d}\E | \hat X_{k+2}^{i} |^{4} )     
     \end{align}
   for all $k=0, 1, \dots, T-1$, $j=l_{k}+1, \dots, d$. By the definition of $V_{k+1}(\hat \varphi)$ and since by assumption $H \in \classlfourone$ and $\hat X_{T+1} = 0$, one can argue recursively that both $\hat X_{k+1}^{j} \Delta \tilde S_{k+1}^{j}$ and $\hat X_{k+1}^{j}$ are in $\classlfourone$.  
   \par
    Furthermore, we have for some $j \in \{ l_{k}+1, \dots, d \}$ at time $k \in \{ 0, 1, \dots, T-1 \}$ for the term 
   \begin{equation}
     b_{k; j}^{0}
     = \condCov{ V_{k+1}(\hat \varphi), S_{k+1}^{j} }
     = \condE{ V_{k+1} S_{k+1}^{j} } - \condE{ V_{k+1} } \condE{ S_{k+1}^{j} }
   \end{equation}
   that $V_{k+1}(\hat \varphi) \in \classltwoone$, $S_{k+1}^{j} \in \classltwoone$ and $V_{k+1}(\hat \varphi) S_{k+1}^{j} \in \classltwoone$ since  $V_{k+1}(\hat \varphi) \in \classlfourone$, $\tilde S_{k+1}^{j} \in \classlfourone$ and by the Cauchy-Schwarz inequality. For the term 
   \begin{equation}
     b_{k; j}^{\varepsilon}
     = \condE{ \varepsilon_{k+1}^{j} S_{k+1}^{j} \hat X_{k+2}^{j} }
   \end{equation}
   we have $S_{k+1}^{j} \hat X_{k+2}^{j} \in \classltwoone$ since $\tilde S_{k+1}^{j} \in \classlfourone$, $\hat X_{k+2}^{j} \in \classlfourone$ and using the Cauchy-Schwarz inequality.
   \par
   So, all random variables in the conditional expectations for the term $b_{k; j}$ are square integrable. Analogously the same holds for the terms $A_{k; j}$ and $D_{k; j, i}$.
  \end{proof}

 We now come to the specification of the electricity market model for our simulation study. To this end, we consider the spot price model \eqref{spotmodel} with two OU factors ($n=2$) $Y_{1}$ the base regime and $Y_{2}$ the spike regime with strong upward moves followed by quick reversion to normal levels and constant seasonality function $\Lambda_{1} = \Lambda_{2} =1$. We set $Y_{1}(0) = Y_{2}(0) = 0.5$, and assume constant volatilities $\sigma^{1}=0.34, \sigma^{2}=0.01$ and mean reversion rates $\lambda_{1}=0.01, \lambda_{2}=0.1$. For the driving L\'evy processes we suppose that $L_{1}$ is a Gamma process where $L_{1}(t)$ has $\Gamma(\gamma^{1} t, \alpha^{1})$-distribution and $L_{2}$ a compound Poisson process with intensity $\gamma^{2}$ and $\exp(\alpha^{2})$-distributed jumps. We set $\gamma^{1} = \gamma^{2} = \alpha^{1} = 1$,  $\alpha^{2} = 0.1$. Both OU-processes are simulated using an Euler Scheme.\footnote{Note that in order to use the Least-squares Monte Carlo metho
 d for calculating conditional expectations for the simulation, we need to simulate $2$-dim. basis functions using both Markov processes $L^{1}$ and $L^{2}$.}
 
Moreover, we set the strike price $K=1.05$ in \eqref{calloption} and $\alpha=1$ in the performance criterium \eqref{eq:quadratic linear risk process new}, which means an equal concern between the risk from market price fluctuations and the cost of liquidity costs.

 We will simulate and analyze two different settings, each with various pairs of futures with different delivery periods as available hedge instruments for the call option. In the first setting we focus on hedging the option with \emph{various combinations of futures that cover the delivery period $[T_{1}^{c},T_{2}^{c}]$ of the option}. To this end we consider three futures $F^{1}, F^{2}, F^{3}$ with delivery periods $[ T_{1}^{F^{1}}, T_{2}^{F^{1}} ], [ T_{1}^{F^{2}}, T_{2}^{F^{2}} ], [ T_{1}^{F^{3}}, T_{2}^{F^{3}} ]$, respectively, where we set $T^c_{1} = T_{1}^{F^{1}} = T_{1}^{F^{2}} = 0.0125$, $T^c_{2} = T_{2}^{F^{2}} = T_{2}^{F^{3}} = 0.1$, $T_{2}^{F^{1}} = T_{1}^{F^{3}} = 0.05$. We consider both the time varying liquidity structure \eqref{TimeLiquidity}, where we set $M_{i}=0.005$, $N_{i}=2M_{i}$, $\delta_i = 0.000001$  and the constant liquidity structure \eqref{ConstantLiquidity}, where we set  $M_{i}=N_{i}=0.01$ for $i=1,2,3$. 
 %Moreover, both in \eqref{TimeLiquidity} and \eqref{ConstantLiquidity} we put $N_{i}=2M_{i}$ and $\delta_i = 0,0001$ for $i=1,2,3$. 
We compute the criteria $ T_{0}(\varphi) $, $\tilde T_{0}(\varphi) $, $L_{0}(\varphi)$, and $C_{0}(\varphi)$ for LRM-strategies $\varphi=(X,Y)$, where $\tilde T_{0}(\varphi) = \E [(C_{T}(\varphi) - C_{0}(\varphi))^{2}]$ is the quadratic hedge criterion, $L_{0}(\varphi) = \E[ \sum_{m=1}^{T} \Delta X_{m+1}^{*} [ S_{m}(\Delta X_{m+1}) - S_{m}(0) ] ]$ the liquidity costs, $ T_{0}(\varphi) = \tilde T_{0}(\varphi)  + L_{0}(\varphi)$ our combined LRM minimization criterion \eqref{eq:quadratic linear risk process}, and $C_{0}(\varphi)  = \E[ H - \sum_{m=1}^{T} (X_{m})^{*} \Delta S_{m} ]$ the cost for a strategy $\varphi$ at time $0$. In Tables~\ref{tab:constant liquidity phenom 1} and  \ref{tab:time varying liquidity phenom 1} the results are displayed for a LRM-strategy $\varphi^{L}=(X^{L},Y^{L})$ with time varying liquidity~\eqref{TimeLiquidity} and constant liquidity~\eqref{ConstantLiquidity}, respectively. In addition, we compute the results with the classical LRM-strategy $\varphi^{C}=(X^{C},Y^{C})$ with zero liquidity costs (i.e., $\varepsilon^{i} = 0$). Recall that the quantity $T_{0}$ is minimized by $\varphi^{L}$ and $\tilde T_{0}$ is minimized by $\varphi^{C}$. For comparison, we use the same trajectories in both cases.

The first observation that can be made is that the hedging costs and the corresponding minimization criterion indeed decrease in the number of available hedge instruments. Also, the initial cost for using the strategy $\varphi^{L}$ is more than using $\varphi^{C}$ since it will cost more to generate the optimal strategy $\varphi^{L}$ under liquidity costs. 
To focus on the hedge performance with two futures that cover the delivery period of the option we consider two examples. In the first one we consider the futures $F^{1}, F^{2}$ with overlapping delivery periods while in the second one, the futures $F^{1}, F^{3}$ (see Figure \ref{fig:example1_2}) have different delivery periods. From Tables \ref{tab:constant liquidity phenom 1} and  \ref{tab:time varying liquidity phenom 1} and by comparing the quantity $T_{0}(\varphi^{L})$ we see that the case with the futures $F^{1}, F^{3}$ performs better since they incur less cost. In Table \ref{tab:time varying liquidity phenom 1} with time-varying liquidity this is due to the fact that $F^{3}$ has shorter delivery period than $F^{2}$ and can be used for hedging longer in time. In Table \ref{tab:constant liquidity phenom 1} we see that in the case with constant liquidity, despite that $F^{2}$ has a delivery period perfectly coinciding with the option $H$ it is better to hedge with the tw
 o hedge instruments $F^{1}$ and $F^{3}$. By looking at the quantity $\tilde T_{0}(\varphi^{C})$ one can observe that also for the classical LRM-strategy under the classical LRM-criterion the futures $F^{1}$ and $F^{3}$ perform better, simply due to the increased dimension of the hedge instruments.
\par
Recall that our quadratic criterion balances low liquidity costs against poor replication. This can be seen for example in Tables \ref{tab:constant liquidity phenom 1} and  \ref{tab:time varying liquidity phenom 1}. Indeed, from our example the futures $F^{1}, F^{3}$ perform better with less cost $\tilde T_{0}(\varphi^{L})$ from market fluctuations but incurring more liquidity cost $L_{0}(\varphi^{L})$ than the futures $F^{1}, F^{2}$.
%Another aspect that we can observe from the two examples in Tables \ref{tab:constant liquidity phenom 1} and  \ref{tab:time varying liquidity phenom 1} is the balance between the quantities $\tilde T_{0}(\varphi^{L})$ and $L_{0}(\varphi^{L})$ of how the quantity $T_{0}(\varphi^{L})$ is minimized by the strategy $\varphi^{L}$ and recall that our quadratic criterion balances low liquidity costs against poor replication. Indeed, the futures $F^{1}, F^{3}$ perform better with less cost $\tilde T_{0}(\varphi^{L})$ from market fluctuations but incurring more liquidity cost $L_{0}(\varphi^{L})$ than the futures $F^{1}, F^{2}$.    
\par
 Note also, that Figure~\ref{fig:example1_1} corresponding to the result for $F^{2}$ in Table \ref{tab:constant liquidity phenom 1} confirms the numerical results of \cite{agliardi.gencay:2014} and \cite{rogers.singh:2010} who find that the optimal strategy under illiquidity is less volatile than the classical one. This is perfectly intuitive since changing position drastically incurs large liquidity cost. In Figure~\ref{fig:example1_2} one can observe that before the start of the delivery periods both futures are used actively, but after entering into the delivery period of $F^{1}$ then almost only the future $F^3$ is used for hedging since $F^3$ is more liquid than $F^{1}$ and expires later.
 
 %In Figure \ref{fig:example1_2} we observe that the future $F^{3}$ is used more for hedging than $F^{1}$ since is more liquid and expires later. 

%where for both we have one short future $F^{1}$ where in the one example we use the future $F^{2}$ and for the other the future $F^{3}$ which is shorter but less correlated with the option $H$ than $F^{2}$. We see by comparing the quantity $T_{0}(\varphi)$ that for hedging the option $H$ in the delivery period is better to use two short futures rather than one long and one short. The reason is that $F^{3}$ has shorter delivery period than $F^{2}$ and can be used for hedging longer in time. Note also that hedging with $F^{1}, F^{3}$ with constant liquidity parameter is even better than hedging with $F^{1}, F^{2}$ with time varying liquidity parameter but the performance of $F^{1}, F^{2}$ becomes much better than with constant liquidity parameter since $T_{0}(\varphi^{L})$ is less. 

In a second setting, we focus on the \emph{trade-off between liquidity costs and hedging performance} appearing in various hedge constellations. To this end we consider three futures $G^{1}, G^{2}, G^{3}$ with delivery periods $[ T_{1}^{G^{1}}, T_{2}^{G^{1}} ], [ T_{1}^{G^{2}}, T_{2}^{G^{2}} ], [ T_{1}^{G^{3}}, T_{2}^{G^{3}} ]$, respectively, and set $T_{1}^{c} = T_{1}^{G^{1}} = T_{1}^{G^{2}} = T_{2}^{G^{3}} = 0.05$, $T_{2}^{c} = T_{2}^{G^{2}} = 0.1$, $T_{2}^{G^{1}} = 0.075$, $T_{1}^{G^{3}} = 0.0125$. Otherwise, the model specifications remain the same as in the first setting above. We consider two examples, with one common future $G^{2}$, which has the same delivery period as the option $H$. From Tables \ref{tab:constant liquidity phenom 2} and \ref{tab:time varying liquidity phenom 2} we can observe that $G^{1}, G^{2}$ performs better than $G^{2}, G^{3}$ according to the quantity $T_{0}(\varphi^{L})$. From $\tilde T_{0}(\varphi^{C})$ we see that this is also the case in the
  classical setting.  This is mostly due to the fact that the future $G^{1}$ expires later than $G^{3}$ and its delivery period lies within the delivery period of the option. Note that by comparing the quantity $T_{0}(\varphi^{L})$ of both examples we observe that in Table \ref{tab:time varying liquidity phenom 2} the difference between them becomes less than in Table \ref{tab:constant liquidity phenom 2}. This is due to the fact that $G^{3}$ is more liquid than $G^{1}$ in the period $[0,0.0125]$ in this case and can be used for hedging at low liquidity cost. Therefore a correct specification of the term-structure of liquidity seems important.  
In Figure \ref{fig:example2} and Figure \ref{fig:example3} we display the strategies for one trajectory in both cases. In Figure~\ref{fig:example3_2} one can actually observe that $G^3$ is the more active hedge instrument in the period $[0,0.0125]$ where it is more liquid than the future $G^2$ in the case with time dependent liquidity.

    \begin{center}
   %\captionof{table}{Your caption here}
   %\label{tab:tab1}
   \scriptsize
   \begin{tabular}{ l | l r | l r | l r | l r }
   Hedging Instruments & $T_{0}(\varphi^{L})$ & $T_{0}(\varphi^{C})$  & $\tilde T_{0}(\varphi^{L})$ & $\tilde T_{0}(\varphi^{C})$ & $L_{0}(\varphi^{L})$ & $L_{0}(\varphi^{C})$ & $C_{0}(\varphi^{L})$ & $C_{0}(\varphi^{C})$   \\ \hline
   $F^{2}$ & 2.19E-3 & 4.79E-2 & 2.03E-3 & 3.40E-4 & 1.56E-4 & 4.76E-2 & 1.09E-2 & 9.29E-3    \\ \hline
   $F^{1}, F^{2}$ & 1.86E-3 & 3.64E-2 & 1.67E-3 & 2.92E-4 & 1.88E-4 & 3.61E-2 & 1.07E-2 & 9.19E-3    \\ \hline
   $F^{1}, F^{3}$ & 1.51E-3 & 1.59E-2 & 1.31E-3 & 2.20E-4 & 2.01E-4 & 1.57E-2 & 1.05E-2 & 8.92E-3    \\ \hline
   \end{tabular}
   \captionof{table}{Simulation results with constant liquidity parameter.}
   \label{tab:constant liquidity phenom 1}
 \end{center}
 %table time varying liquidity phenom 1
\begin{center}
   \scriptsize
   \begin{tabular}{ l | l r | l r | l r | l r}
   Hedging Instruments & $T_{0}(\varphi^{L})$ & $T_{0}(\varphi^{C})$  & $\tilde T_{0}(\varphi^{L})$ & $\tilde T_{0}(\varphi^{C})$ & $L_{0}(\varphi^{L})$ & $L_{0}(\varphi^{C})$ & $C_{0}(\varphi^{L})$ & $C_{0}(\varphi^{C})$   \\ \hline
   $F^{2}$ & 1.63E-3 & 4.11E-2 & 1.49E-3 & 3.40E-4 & 1.40E-4 & 4.08E-2 & 1.05E-2 & 9.29E-3   \\ \hline
   $F^{1}, F^{2}$ & 1.56E-3 & 3.58E-2 & 1.35E-3 & 2.92E-4 & 2.10E-4 & 3.55E-2 & 1.04E-2 & 9.19E-3    \\ \hline
   $F^{1}, F^{3}$ & 7.09E-4 & 1.28E-2 & 4.50E-4 & 2.20E-4 & 2.59E-4 & 1.26E-2 & 9.66E-3 & 8.92E-3    \\ \hline
   \end{tabular}
   \captionof{table}{Simulation results with time varying liquidity parameter.}
   \label{tab:time varying liquidity phenom 1}
 \end{center}
 %table constant liquidity phenom 2
    \begin{center}
   %\captionof{table}{Your caption here}
   %\label{tab:tab1}
   \scriptsize
   \begin{tabular}{ l | l r | l r | l r | l r }
   Hedging Instruments & $T_{0}(\varphi^{L})$ & $T_{0}(\varphi^{C})$  & $\tilde T_{0}(\varphi^{L})$ & $\tilde T_{0}(\varphi^{C})$ & $L_{0}(\varphi^{L})$ & $L_{0}(\varphi^{C})$ & $C_{0}(\varphi^{L})$ & $C_{0}(\varphi^{C})$   \\ \hline
   $G^{2}$ & 3.22E-3 & 2.30E-2 & 2.99E-3 & 7.75E-4 & 2.28E-4 & 2.23E-2 & 1.60E-2 & 1.41E-2    \\ \hline
   $G^{1}, G^{2}$ & 2.33E-3 & 8.03E-3 & 2.06E-3 & 5.21E-4 & 2.68E-4 & 7.51E-3 & 1.55E-2 & 1.39E-2    \\ \hline
   $G^{2}, G^{3}$ & 2.95E-3 & 1.52E-2 & 2.69E-3 & 7.12E-4 & 2.55E-4 & 1.45E-2 & 1.58E-2 & 1.40E-2    \\ \hline
   \end{tabular}
   \captionof{table}{Simulation results with constant liquidity parameter.}
   \label{tab:constant liquidity phenom 2}
 \end{center}
 %table time varying liquidity phenom 2
\begin{center}
   \scriptsize
   \begin{tabular}{ l | l r | l r | l r | l r}
   Hedging Instruments & $T_{0}(\varphi^{L})$ & $T_{0}(\varphi^{C})$  & $\tilde T_{0}(\varphi^{L})$ & $\tilde T_{0}(\varphi^{C})$ & $L_{0}(\varphi^{L})$ & $L_{0}(\varphi^{C})$ & $C_{0}(\varphi^{L})$ & $C_{0}(\varphi^{C})$   \\ \hline
   $G^{2}$ & 1.66E-3 & 1.45E-2 & 1.49E-3 & 7.75E-4 & 1.69E-4 & 1.37E-2 & 1.50E-2 & 1.41E-2    \\ \hline
   $G^{1}, G^{2}$ & 1.32E-3 & 4.64E-3 & 1.13E-3 & 5.21E-4 & 1.92E-4 & 4.12E-3 & 1.47E-2 & 1.39E-2    \\ \hline
   $G^{2}, G^{3}$ & 1.63E-3 & 1.25E-2 & 1.39E-3 & 7.12E-4 & 2.39E-4 & 1.18E-2 & 1.49E-2 & 1.40E-2    \\ \hline
   \end{tabular}
   \captionof{table}{Simulation results with time varying liquidity parameter.}
   \label{tab:time varying liquidity phenom 2}
 \end{center}
\begin{figure}[h]
\begin{subfigure}[c]{1.0\textwidth} %0.5\textwidth
\centering
\includegraphics[width=0.7\textwidth, height=9\baselineskip]{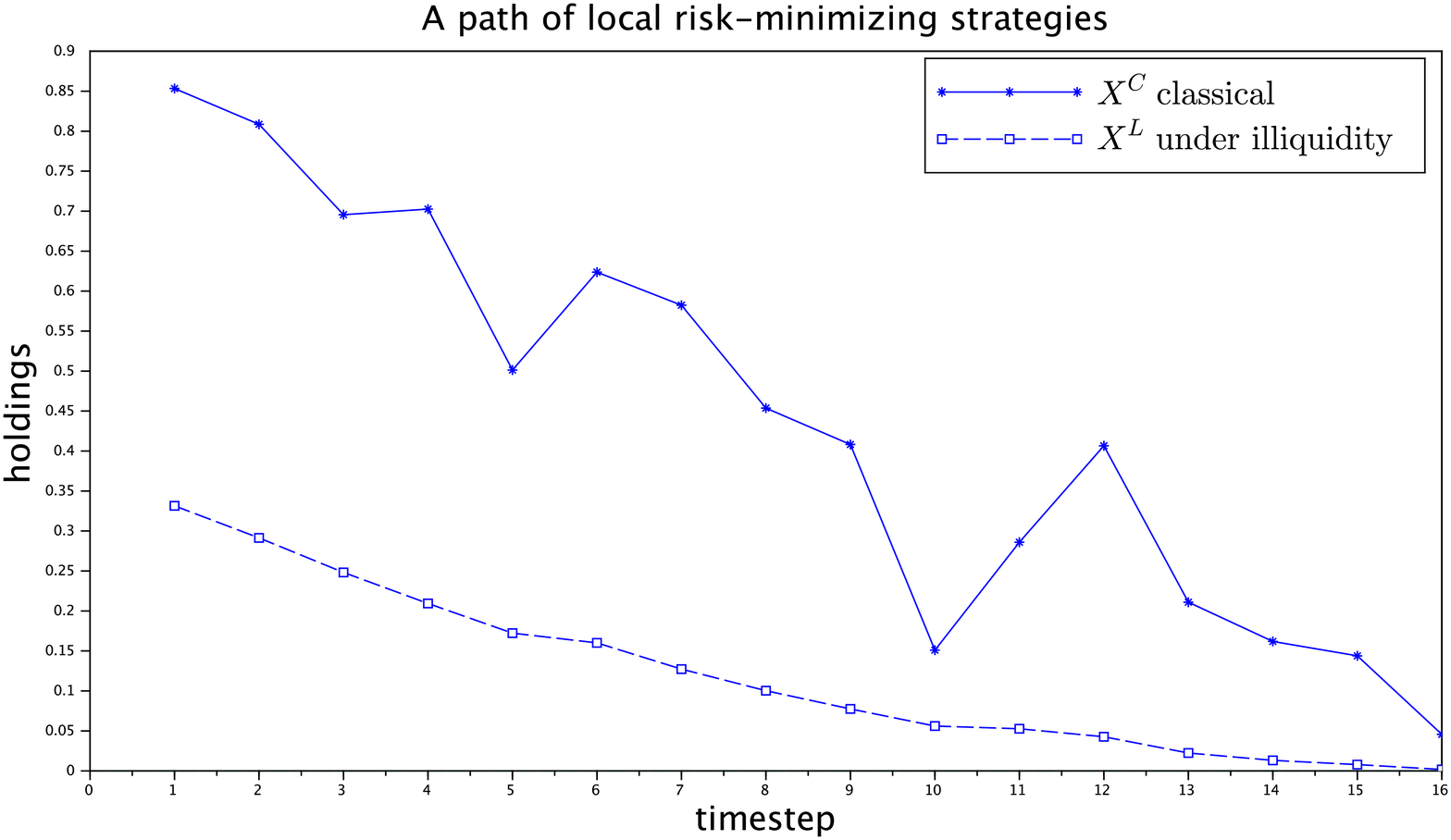}%1.1\textwidth
\subcaption{Hedging with only the Future $F^{2}$ which has the same delivery period as the claim $H$. The hedging strategy $X^{C}$ corresponds to the classical case without liquidity cost and $X^{L}$ to the case with constant liquidity structure (\ref{ConstantLiquidity}) with parameters $M_{2}=N_{2}=0.01$. Observe that the optimal LRM-strategy $X^{L}$ under illiquidity is less volatile than the classical LRM-strategy $X^{C}$.  }
\label{fig:example1_1}
\end{subfigure}
\begin{subfigure}[c]{1.0\textwidth}%0.5\textwidth
\centering
\includegraphics[width=0.7\textwidth, height=9\baselineskip]{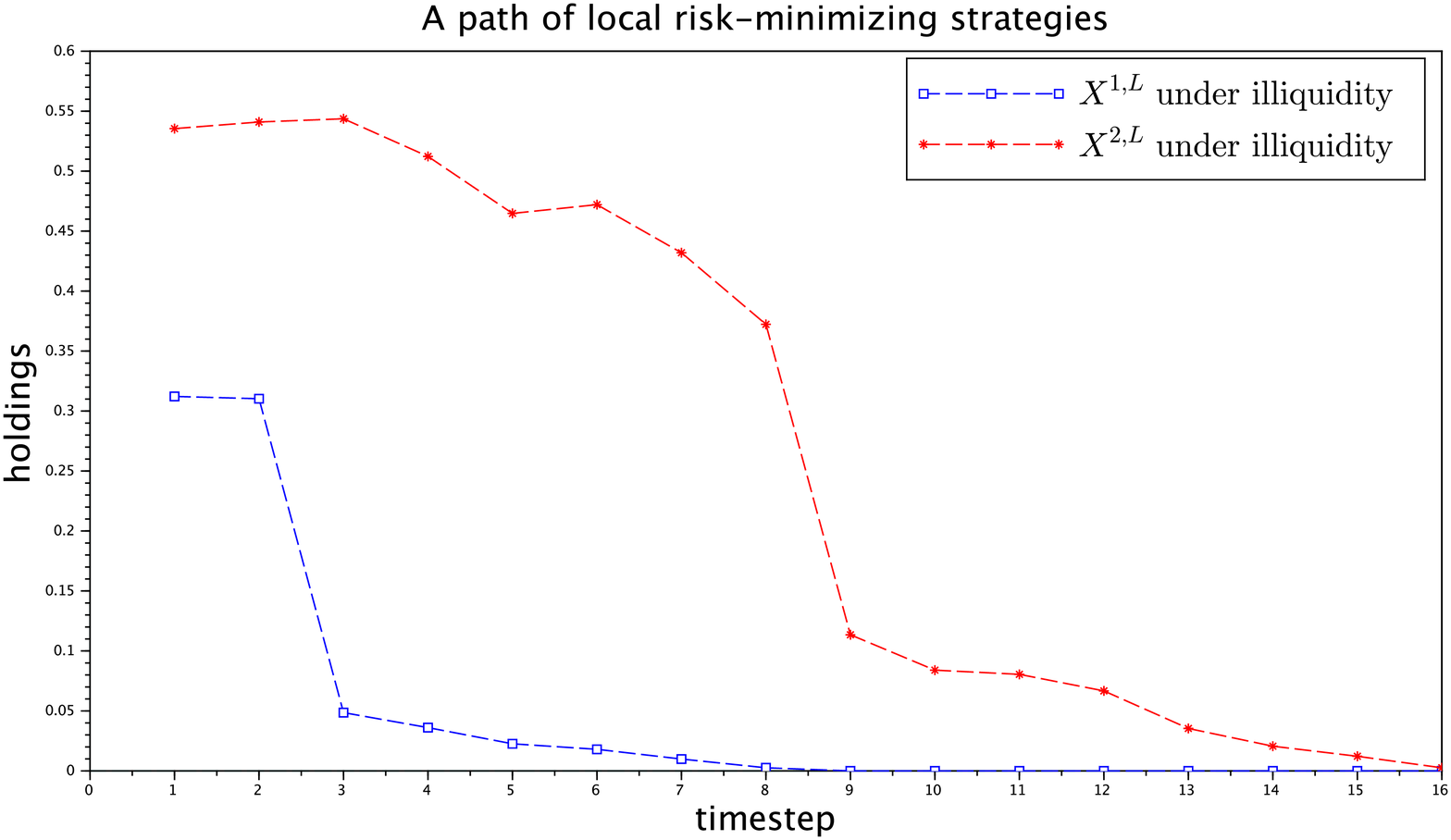}%1.1\textwidth
\subcaption{Hedging with the two futures $F^{1}$ and $F^{3}$ with consecutive delivery periods which together cover the delivery period of the claim $H$. Both futures have a time-varying liquidity structure (\ref{TimeLiquidity}) with parameters $M_{1}=M_{3}=0.005, N_{1}=2M_{1}, N_{3}=2M_{3}$. The optimal LRM-strategy $X^{1, L}$ under illiquidity corresponds to the future $F^{1}$ and $X^{2, L}$ to the future $F^{3}$ which is more liquid and expires later than $F^{1}$. $F^{3}$ is used more actively as can be observed from the plot. In the delivery period both futures become very illiquid and thus a rapid drop in the holdings can be observed. }
\label{fig:example1_2}
\end{subfigure}
\caption{Comparison of the sample path of optimal LRM-strategies under different liquidity structures and for different hedge instruments. All plots based on the same realization of the underlying.}
\label{fig:example1}
\end{figure}
 %Figure 2
\begin{figure}[h]
\begin{subfigure}[c]{1.0\textwidth}%0.5\textwidth
\centering
\includegraphics[width=0.7\textwidth, height=9\baselineskip]{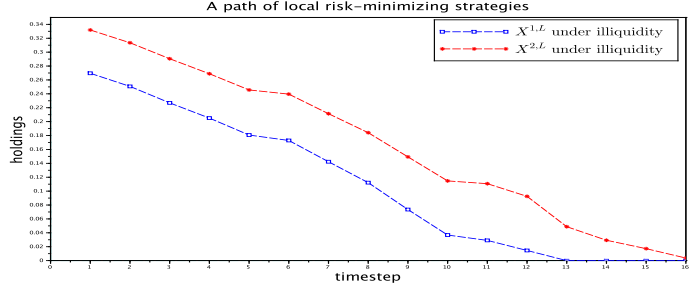}%1.1\textwidth
\subcaption{Hedging with the two Futures $G^{1}, G^{2}$ with constant liquidity structure (\ref{ConstantLiquidity}) with parameters $M_{i}=N_{i}=0.01$ for $i=1, 2$ .}
\label{fig:example2_1}
\end{subfigure}
\begin{subfigure}[c]{1.0\textwidth}%0.5\textwidth
\centering
\includegraphics[width=0.7\textwidth, height=9\baselineskip]{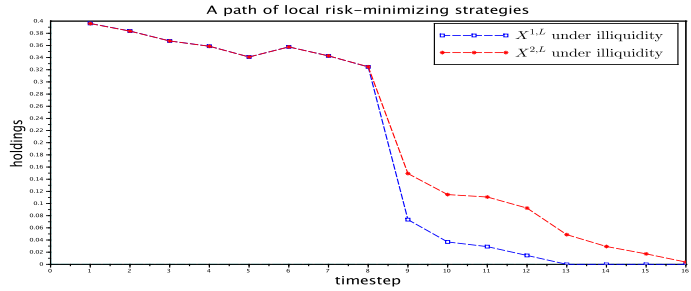}%1.1\textwidth
\subcaption{Hedging with two instruments using the Futures $G^{1}, G^{2}$ with time-varying liquidity structure (\ref{TimeLiquidity}) with parameters $M_{i}=0.005, N_{i}=2M_{i}$ for $i=1, 2$. The sudden drop in the holdings occurs when entering the delivery period where the futures are very illiquid.}
\label{fig:example2_2}
\end{subfigure}
\caption{Comparison of the sample path of optimal LRM-strategies under different liquidity structures but with the same hedge instruments. The two futures have overlapping delivery periods starting together but $G^{1}$ expires earlier. The delivery period of $G^{2}$ is the same as the one of the claim $H$. The LRM-strategies $X^{1, L}, X^{2, L}$ under illiquidity correspond to the future hedge instruments $G^{1}, G^{2}$ respectively.}
\label{fig:example2}
\end{figure}
%Figure 3
\begin{figure}[h]
\begin{subfigure}[c]{1.0\textwidth}%0.5\textwidth
\centering
\includegraphics[width=0.7\textwidth, height=9\baselineskip]{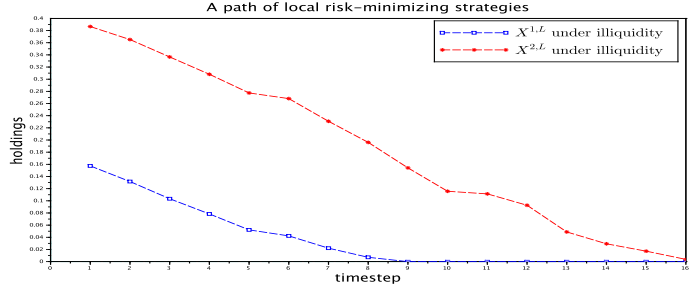}%1.1\textwidth
\subcaption{Hedging with the two Futures $G^{2}, G^{3}$ with constant liquidity structure (\ref{ConstantLiquidity}) with parameters $M_{i}=N_{i}=0.01$ for $i=2, 3$.}
\label{fig:example3_1}
\end{subfigure}
\begin{subfigure}[c]{1.0\textwidth}%0.5\textwidth
\centering
\includegraphics[width=0.7\textwidth, height=9\baselineskip]{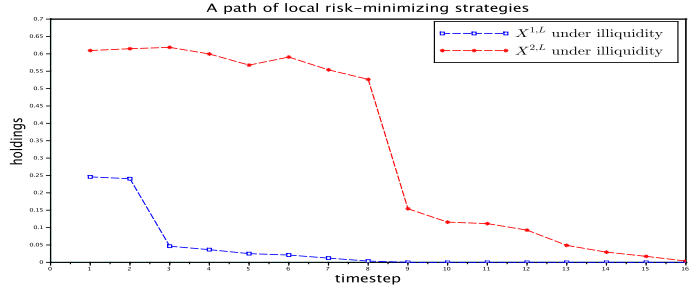}%1.1\textwidth
\subcaption{Hedging with the two Futures $G^{2}, G^{3}$ with time-varying liquidity structure (\ref{TimeLiquidity}) with parameters $M_{i}=0.005, N_{i}=2M_{i}$ for $i=2, 3$.}
\label{fig:example3_2}
\end{subfigure}
\caption{Comparison of the sample path of optimal LRM-strategies under different liquidity structures but with the same hedge instruments. The two futures have consecutive delivery periods with the one of $G^{3}$ starting earlier. The delivery period of $G^{2}$ and the claim $H$ coincide. The optimal LRM-strategy $X^{1, L}$ under illiquidity corresponds to the hedge instrument $G^{3}$ and  $X^{2, L}$ to $G^{2}$.}
\label{fig:example3}
\end{figure}

   %In the delivery period $[\tilde T_{1}, \tilde T_{2}]$ the future $F^{1}$ is more correlated with the option %$H$ than the future $F^{2}$. This can be seen for example in figure \ref%{fig:example1Table2And3DeltasIlliq} where the LOB is constant. Increasing the liquidity level for the %second future, then we see from figure \ref{fig:example1Table3DeltasIlliq} how the hedging strategy $%\Delta X^{2, L}$ decreases and $\Delta X^{1, L}$ increases since $F^{1}$ becomes cheaper than $F^%{2}$. On the other hand, one can observe from figure \ref{fig:example2Table2And3DeltasIlliq} that %despite of the correlation between $F^{1}$ and $H$ in $[\tilde T_{1}, \tilde T_{2}]$ the investor hedges %more with the second future since is much cheaper. This is due to the decreasing liquidity function $%\varepsilon^{2}$ for the future $F^{2}$ in the period $[\tilde T_{1}, \tilde T_{2}]$. In figure \ref%{fig:example2Table3DeltasIlliq} $F^{2}$ becomes more expensive and so $X^{2, L}$ becomes smoo
% ther %than before but still the investor hedges more with $F^{2}$.           

%Conclusion
\section{Conclusion}
\label{sec:conclusion}

In an arbitrage-free model framework, this paper has presented a new quadratic hedging criterion that targets at minimizing the risk against random fluctuations of the underlying stock price while simultaneously incurring low liquidity costs. It extends the quadratic local-risk minimization approach of \cite{schweizer:1988} in the spirit of \cite{rogers.singh:2010} and \cite{agliardi.gencay:2014}. It is mathematically tractable enough to allow for computable formulae. Under mild conditions, the optimization problem can be solved in closed-form. Furthermore, by embedding a multi-dimensional price process with different maturities in our setting it is possible to consider as one possible application the hedging of an Asian-style option in an electricity exchange using a variety of futures. In a simulation study we analyze hedge performance and cost under various pairs of futures with different delivery periods and liquidity levels, allowing us to investigate the tradeoff between hedge performance and liquidity cost.
%Our numerical results are in line with the findings in \cite{rogers.singh:2010} and \cite{agliardi.gencay:2014}\nils{Dann muessen wir auch sagen was die da finden. Koennen ihn aber auch raus lassen.}. \panos{vielleicht: ... where optimal strategies under illiquidity are less volatile than classical ones.}
%Also, as expected, the hedging costs decrease in the number of available hedge instruments. 
%\nils{Diesen naechsten Satz raus? Steht doch oben schon?}\panos{ja.. steht doch oben einigermassen. Also den naechsten Satz koennen wir weglassen. }Our hedging criterion balances low liquidity costs against poor replication.}

\section{Appendix}
\label{appendix}
\begin{proof}[Proof of Lemma \ref{lem:cost process martingale}]
 The arguments follow those in the proof of Lemma 1 in \cite{lamberton.pham.schweizer:1998}.
   \par Let $\varphi=(X,Y)$ be a LRM-strategy under illiquidity and fix some $k \in \{ 0, 1, \dots, T-1 \}$. Assuming that $C(\varphi)$ is not a martingale, we can choose a local perturbation $\varphi ' = (X',Y')$ of $\varphi$ at time $k$ by defining $X':=X$ and only modifying the cash holding $Y'$ at time $k$, by adding the conditional expectation of the incremental cost at time $k$ to $Y$, 
   \begin{equation}
     Y_{k}^{'} := \condE{ C_{T}(\varphi) - C_{k}(\varphi) } + Y_{k} \,.
   \end{equation} 
   This implies that $\condE{ C_{T}(\varphi') - C_{k}(\varphi') } = 0$ and 
   $\condVar{ C_{T}(\varphi') - C_{k}(\varphi') } = \condVar{C_{T}(\varphi) - C_{k}(\varphi) }$. Since $\E[X^2]= \Var [X] + (\E[X])^{2}$ for a random variable $X$, one can conclude that using the strategy $\varphi'$ the risk process becomes less, that is,
   \begin{equation}
     R_{k}(\varphi')
       \leq R_{k}(\varphi)  \,.
   \end{equation}
   Since $X':=X$, the liquidity costs of $\varphi'$ and $\varphi$ equal. This implies, 
   \begin{equation}
     T_{k}^{\alpha}(\varphi')
       \leq T_{k}^{\alpha}(\varphi)  \,.
   \end{equation}
   By the fact that $\varphi$ is a LRM-strategy under illiquidity, we must have equality on $T_{k}^{\alpha}$ which implies equality on $R_{k}$ i.e., $R_{k}(\varphi') = R_{k}(\varphi)$. So, the cost process $C(\varphi)$ must be a martingale.
\end{proof}

\begin{proof}[Proof of Lemma \ref{lem:quadratic linear risk process modified}]
   As in \cite{lamberton.pham.schweizer:1998} (see proof of Proposition 2), by using Lemma \ref{lem:cost process martingale} and the fact that 
   \begin{equation}
     \condE{ C_{T}(\varphi') - C_{k}(\varphi') }
     = \Delta C_{k+1}(\varphi'),
   \end{equation}
   which follows from the martingale property of $C(\varphi)$, one can conclude that 
   \begin{equation}
     R_{k}(\varphi ')
         =  \condE{ R_{k+1}(\varphi) }
             + \condE{ (\Delta C_{k+1}(\varphi '))^{2} }  \,.
   \end{equation}
   Furthermore since $\varphi '$ is a local perturbation of $\varphi$ at time $k$, we have 
   \begin{align}
     & \condE{ (X'_{k+2} - X'_{k+1})^{*} [ S_{k+1}(X'_{k+2} - X'_{k+1}) - S_{k+1}(0) ] }
      \nonumber \\ & \quad 
     = \condE{ (X_{k+2} - X'_{k+1})^{*} [ S_{k+1}(X_{k+2} - X'_{k+1}) - S_{k+1}(0) ] }
   \end{align}
   and the claim follows.
\end{proof}

\begin{proof}[Proof of Proposition \ref{prop:minimizing variance}]
  The proof follows the steps in the proof of Proposition 2 in \cite{lamberton.pham.schweizer:1998}. 
   \par Let us first show the $``\Leftarrow"$ direction of the proof. We want to show that $\varphi=(X,Y)$ is a LRM-strategy under illiquidity, according to Definition \ref{defi:local risk minimizing strategy under illiquidity}. So, fix some $k \in \{ 0, 1, \dots, T-1 \}$ and let $\varphi '=(X',Y')$ be a local perturbation of $\varphi$ at time $k$. 
   \par Since property \ref{prop:1} holds and $\varphi '$ a local perturbation of $\varphi$ at time $k$ then by Lemma \ref{lem:quadratic linear risk process modified} we have the equality
    \begin{align}
       T_{k}^{\alpha}(\varphi ')
         = & \condE{ R_{k+1}(\varphi) }
             + \condE{ (\Delta C_{k+1}(\varphi '))^{2} }
            \nonumber \\  
            &+ \alpha \condE{ (X_{k+2} - X'_{k+1})^{*} [ S_{k+1}(X_{k+2} - X'_{k+1}) - S_{k+1}(0) ] }
     \end{align}
     Moreover, from the definition of the conditional variance we have
     \begin{equation}
       \condE{ (\Delta C_{k+1}(\varphi '))^{2} }
       \geq \condVar{ \Delta C_{k+1}(\varphi ') }
     \end{equation}
     and so we can estimate
     \begin{align}
     T_{k}^{\alpha}(\varphi ')
         \geq \, & \condE{ R_{k+1}(\varphi) }
             + \condVar{ \Delta C_{k+1}(\varphi ') }
            \nonumber \\ 
            & + \alpha \condE{ (X_{k+2} - X'_{k+1})^{*} [ S_{k+1}(X_{k+2} - X'_{k+1}) - S_{k+1}(0) ] } \,.
   \end{align}
   Since $\varphi '$ a local perturbation of $\varphi$ at time $k$ then $X'_{k+2} = X_{k+2}$ and $Y'_{k+1} = Y_{k+1}$ and so we get
   \begin{align}
     \condVar{ \Delta C_{k+1}(\varphi ') }
     = \condVar{ C_{k+1}(\varphi ') }
     &= \condVar{ V_{k+1}(\varphi') - (X'_{k+1})^{*} \Delta S_{k+1} }
     \nonumber \\
     &= \condVar{ V_{k+1}(\varphi) - (X'_{k+1})^{*} \Delta S_{k+1} }
   \end{align}
 and we can conclude that
   \begin{align}
     T_{k}^{\alpha}(\varphi ')
         \geq \, & \condE{ R_{k+1}(\varphi) }
             + \condVar{ V_{k+1}(\varphi) - (X'_{k+1})^{*} \Delta S_{k+1} }
            \nonumber \\ 
            & + \alpha \condE{ (X_{k+2} - X'_{k+1})^{*} [ S_{k+1}(X_{k+2} - X'_{k+1}) - S_{k+1}(0) ] } \,.
   \end{align}
   Furthermore, since \ref{prop:2} holds, then 
    \begin{align}
    \label{eq:propProofEq1}
     T_{k}^{\alpha}(\varphi ')
         \geq \, & \condE{ R_{k+1}(\varphi) }
             + \condVar{ V_{k+1}(\varphi) - (X_{k+1})^{*} \Delta S_{k+1} }
            \nonumber \\ 
            & + \alpha \condE{ (X_{k+2} - X_{k+1})^{*} [ S_{k+1}(X_{k+2} - X_{k+1}) - S_{k+1}(0) ] } \,.
   \end{align}
On the other hand, we have by definition (see Equation (\ref{eq:quadratic linear risk process new}) )
   \begin{equation}
       T_{k}^{\alpha}(\varphi)
         = R_{k}(\varphi)
             + \alpha \condE{ \Delta X_{k+2}^{*} [ S_{k+1}(\Delta X_{k+2}) - S_{k+1}(0) ] }  \,.
   \end{equation}
   Since $C(\varphi)$ is a martingale, we get the representation (\ref{eq:cost process martingale}) for the risk process $R_{k}(\varphi)$. So we can conclude that
   \begin{align}
       \label{eq:propProofEq2}
       T_{k}^{\alpha}(\varphi)
         = & \condE{ R_{k+1}(\varphi) }
             + \condVar{ \Delta C_{k+1}(\varphi) }
            \nonumber \\ 
            & + \alpha \condE{ (X_{k+2} - X_{k+1})^{*} [ S_{k+1}(X_{k+2} - X_{k+1}) - S_{k+1}(0) ] }
     \end{align}
     Finally, since (\ref{eq:propProofEq1}) and (\ref{eq:propProofEq2}) hold then $T_{k}^{\alpha}(\varphi ') \geq T_{k}^{\alpha}(\varphi)$ and this shows the $``\Leftarrow"$ direction of the proof.
   \par Now,  assuming that $\varphi$ is a LRM-strategy under illiquidity i.e., $T_{k}^{\alpha}(\varphi ') \geq T_{k}^{\alpha}(\varphi)$ for any local perturbation $\varphi '$ of $\varphi$ at time $k$, we will show the $``\Rightarrow"$ direction of the proof. Property \ref{prop:1} holds from Lemma~\ref{lem:cost process martingale}. So it remains to show Property \ref{prop:2}.
   \par Since $C(\varphi)$ is a martingale and $\varphi '$ a local perturbation of $\varphi$ at time $k$, then from Lemma \ref{lem:quadratic linear risk process modified} we know that equation (\ref{eq:quadratic linear risk process modified}) holds. On the other hand, since (\ref{eq:propProofEq2}) holds (from the martingale property of $C(\varphi)$) then from the fact that $T_{k}^{\alpha}(\varphi ') \geq T_{k}^{\alpha}(\varphi)$ we have
   \begin{align}
      & \condE{ R_{k+1}(\varphi) }
             + \condE{ ( \Delta C_{k+1}(\varphi ') )^{2} }
            \nonumber \\
            & \quad + \alpha \condE{ (X_{k+2} - X'_{k+1})^{*} [ S_{k+1}(X_{k+2} - X'_{k+1}) - S_{k+1}(0) ] } 
      \nonumber \\ 
      &\geq \condE{ R_{k+1}(\varphi) }
             + \condVar{ \Delta C_{k+1}(\varphi) }
            \nonumber \\ 
            & \quad + \alpha \condE{ (X_{k+2} - X_{k+1})^{*} [ S_{k+1}(X_{k+2} - X_{k+1}) - S_{k+1}(0) ] } 
   \end{align}
   and from the definition of the conditional variance we can conclude that
   \begin{align}
      & \condVar{ \Delta C_{k+1}(\varphi ') }
              + ( \condE{ \Delta C_{k+1}(\varphi ')  } )^{2}
           \nonumber \\ 
           & \quad + \alpha \condE{ (X_{k+2} - X'_{k+1})^{*} [ S_{k+1}(X_{k+2} - X'_{k+1}) - S_{k+1}(0) ] } 
      \nonumber \\ 
      &\geq 
             \condVar{ \Delta C_{k+1}(\varphi) }
             + \alpha \condE{ (X_{k+2} - X_{k+1})^{*} [ S_{k+1}(X_{k+2} - X_{k+1}) - S_{k+1}(0) ] } 
   \end{align}
    for all $X'_{k+1}$ and $Y'_{k}$. Fixing $X'_{k+1}$ and choosing $Y'_{k}$ as in the proof of Lemma \ref{lem:cost process martingale} the inequality still holds and the liquidity costs remain unchanged. Since this choice gives us $\condE{ \Delta C_{k+1}(\varphi ') }= 0$ (as in the proof of Lemma \ref{lem:cost process martingale}) and since $\varphi '$ a local perturbation of $\varphi$ at time $k$, we get the inequality  
    \begin{align}
       & \condVar{ V_{k+1}(\varphi) - (X'_{k+1})^{*} \Delta S_{k+1} }
       \nonumber \\
       & \quad \quad + \alpha \condE{ (X_{k+2} - X'_{k+1})^{*} [ S_{k+1}(X_{k+2} - X'_{k+1}) - S_{k+1}(0) ] }
       \nonumber \\ 
       \geq & \condVar{ V_{k+1}(\varphi) - (X_{k+1})^{*} \Delta S_{k+1} }
       \nonumber \\
       & \quad \quad + \alpha \condE{ (X_{k+2} - X_{k+1})^{*} [ S_{k+1}(X_{k+2} - X_{k+1}) - S_{k+1}(0) ] } \,.
     \end{align}
This shows that Property \ref{prop:2} holds and the proof is completed.
\end{proof}

\bibliographystyle{newapa}%dinat
\bibliography{hedgingUnderIlliquidity}

\end{document}